\documentclass[letterpaper]{amsart}

\usepackage[utf8]{inputenc}
\usepackage{microtype}
\usepackage{verbatim}
\usepackage{amssymb,amsfonts,amstext,amsmath,upref}
\usepackage{multirow}
\usepackage{multicol}
\usepackage{subfig}
\usepackage{tikz}

\DeclareMathOperator{\GF}{GF}
\newcommand{\F}{\ensuremath{\mathbb F}}
\newcommand{\FF}{\ensuremath{\mathbb F}}
\newcommand{\Fp}{\ensuremath{\FF_p}}
\newcommand{\ZZ}{\ensuremath{\mathbb Z}}
\newcommand{\QQ}{\ensuremath{\mathbb{Q}}}
\DeclareMathOperator{\Reslt}{Res}
\DeclareMathOperator{\vlog}{vlog}
\DeclareMathOperator{\Prb}{Pr}
\DeclareMathOperator{\Norm}{Norm}
\newcommand{\norm}[1]{{\|#1\|}}
\DeclareMathOperator{\cff}{coeff}
\DeclareMathOperator{\LLL}{LLL}
\newcommand{\mc}[3]{\multicolumn{#1}{#2}{#3}}

\newcommand{\coeffs}[1]{#1}

\newcommand{\JLSVi}{{JLSV\textsubscript{1}}}
\newcommand{\JLSVii}{{JLSV\textsubscript{2}}}

\usepackage{hyperref}
\hypersetup{
   pdftitle={Faster individual discrete logarithms},    %title
   pdfauthor={Aurore Guillevic},            % author
   pdfsubject={Discrete Logarithms in Finite Fields},           % subject of the document
   colorlinks=true,         % false: boxed links; true: colored links
   linkcolor=black,         % color of internal links
   citecolor=black,         % color of links to bibliography
   filecolor=black,         % color of file links
   urlcolor=black,          % color of external links
}

\usepackage[onelanguage,ruled,vlined,linesnumbered]{algorithm2e}

\newtheorem{theorem}{Theorem}[section]
\newtheorem{lemma}[theorem]{Lemma}
\newtheorem{proposition}[theorem]{Proposition}
\newtheorem{corollary}[theorem]{Corollary}

\theoremstyle{definition}
\newtheorem{definition}[theorem]{Definition}
\newtheorem{example}[theorem]{Example}

\theoremstyle{remark}
\newtheorem{remark}[theorem]{Remark}

\numberwithin{equation}{section}

\title[Faster individual discrete logarithms in finite fields]
      {Faster individual discrete logarithms in finite fields of composite
        extension degree}

\author{Aurore Guillevic}
\address{%
  Inria Nancy--Grand Est,
  \'{E}quipe Caramba,
  615 rue du jardin botanique,
  CS 20101,
  54603 Villers-lès-Nancy Cedex,
  France}
\email{aurore.guillevic@inria.fr}
\urladdr{\url{https://members.loria.fr/AGuillevic}}
\thanks{}
\subjclass[2010]{Primary 11T71: Cryptography}
\date{September 16, 2018.
This document is the author's version.
  First published in Mathematics of Computation, 2018, 
  published by the American Mathematical Society.
  Publisher's version available online September 6, 2018,
  at \url{https://doi.org/10.1090/mcom/3376}.
  See \url{http://www.ams.org/publications/authors/ctp}
  about AMS copyright.
}

\dedicatory{}
\keywords{Finite field, discrete logarithm, number field sieve,
  function field sieve, individual logarithm.}

\begin{document}
%    Abstract is required.
\begin{abstract}
Computing discrete logarithms in finite fields is a main concern in
cryptography. The best algorithms in large and medium
characteristic fields (e.g., $\GF(p^2)$, $\GF(p^{12})$) are the Number Field Sieve
and its variants (special, high-degree, tower). The best algorithms in small
characteristic finite fields (e.g., $\GF(3^{6 \cdot 509})$) are the Function
Field Sieve, Joux's algorithm, and the quasipolynomial-time algorithm.
The last step of this family of algorithms is the individual logarithm
computation. It computes a smooth decomposition of a given target in
two phases: an initial splitting, then a descent tree. 
While new improvements have been made to reduce the complexity of the
dominating relation collection and linear algebra steps,
resulting in a smaller factor basis (database of known logarithms
of small elements), the last step remains at the same level of difficulty. Indeed,
we have to find a smooth decomposition of a typically large element in
the finite field.
This work improves the initial splitting phase and applies to any
nonprime finite field. It is very efficient when the extension degree
is composite. It exploits the proper subfields, resulting in a much
more smooth decomposition of the target. This leads to a new trade-off
between the initial splitting step and the descent step in small characteristic.
Moreover it reduces the width and the height of the subsequent descent
tree.
\end{abstract}

\maketitle

\section{Introduction}

This work is interested in improving the last step of discrete
logarithm (DL) computations in nonprime finite fields.
The discrete logarithm instances that we target come from
Diffie-Hellman (DH) \cite{DifHel76} key-exchange,
or from pairing-based cryptography. In the latter case, the security relies on
the hardness of computing discrete logarithms in two groups:
the group of points of a particular elliptic curve defined over a
finite field, and
a small extension of this finite field
(in most of the cases of degree 2, 3, 4, 6, or 12).

The finite fields fall in three groups: small, medium and
large characteristic finite fields, corresponding to the respective size of the
characteristic $p$ compared to the total size $Q=p^n$ of the finite field.
This is formalized with the $L$ notation:
\begin{equation}
  \label{eq:L-notation}
  L_{Q}[\alpha,c] = e^{(c+o(1))(\log Q)^{\alpha}(\log \log
    Q)^{1-\alpha}}~,\mbox{ where } Q = p^n,~\alpha \in [0,1],~ c
  \neq 0.
\end{equation}
Small, medium and large characteristic correspond to $\alpha < 1/3$, $1/3 <
\alpha < 2/3$, and $\alpha >2/3$ respectively. The boundary cases are
$\alpha=1/3$ and $\alpha=2/3$.
In large characteristic, that is $p = L_Q[\alpha, c]$ where $\alpha
> 2/3$, the Number Field Sieve (NFS) \cite{Gordon93,Schirokauer93,JouLer02} provides the
best expected running time: in $L_Q[1/3, (64/9)^{1/3}\approx 1.923]$
and was used in the latest record computations in a 768-bit prime
field~\cite{EC:KDLPS17}.
Its special variant in expected running time
$L_Q[1/3, (32/9)^{1/3}\approx 1.526]$
was used to break a 1024-bit trapdoored prime field~\cite{EC:FGHT17}. 
In 2015 and 2016, the Tower-NFS construction of Schirokauer was revisited for
prime fields \cite{AC:BarGauKle15}, then
 Kim, Barbulescu and Jeong improved it 
for nonprime finite fields $\F_{p^n}$ where the extension degree $n$
is composite \cite{C:KimBar16,PKC:KimJeo17}, and used the name
Extended TNFS algorithm. To avoid a confusion due to the profusion of names
denoting variants of the same algorithm, in this paper we will use TNFS
as a generic term to denote the family
of all the variants of NFS that use a tower of number fields. 
Small characteristic means $p=L_Q[\alpha,c]$ where $\alpha < 1/3$. The first
$L[1/3]$ algorithm was proposed by Coppersmith, and generalized as the Function
Field Sieve \cite{Adleman94,AdlHua99}.
 
The NFS and FFS algorithms are made of four phases: polynomial selection
(two polynomials are chosen), relation collection where relations
between small elements are obtained, linear algebra
(computing the kernel of a huge sparse matrix over an auxiliary large
prime finite field), and individual discrete logarithm computation. 
In this work, we improve this last step.
All the improvements of NFS, FFS, and related variants since the 90's decrease the
size of the factor basis, that is, the database of known discrete
logarithms of \emph{small} elements obtained after the linear algebra
step, \emph{small} meaning an element represented by a polynomial of
small degree (FFS), resp., an element whose pseudonorm is small (NFS).
The effort required in the individual discrete logarithm step increases:
one needs to find a decomposition of a given target into \emph{small} elements, to be able
to express its discrete logarithm in terms of already known logarithms
of elements in the factor basis, while the factor basis has decreased
at each major improvement.
In characteristic 2 and 3 where the extension degree is composite, obtaining the
discrete logarithms of the factor basis elements can be done in polynomial time.
The individual discrete logarithm is the most costly part, in quasi-polynomial
-time in the most favorable cases \cite{EC:BGJT14,C:GraKleZum14}.
In practice, the record computations
\cite{NMBRTHY:GGMZ13,NMBRTHY:Joux13,AMOR14,AMOR15,NMBRTHY:JouPie14,NMBRTHY:ACCMORR16} 
 implement hybrid algorithms made of Joux's  $L[1/4]$
algorithm \cite{SAC:Joux13}, and the individual discrete logarithm is computed with
a continued fraction descent, then a classical descent, a QPA descent, and a
Gr{\"o}bner basis descent, or a \emph{powers-of-two} descent algorithm
(a.k.a.~ zig-zag descent) \cite{EPRINT:GraKleZum14b,NMBRTHY:GraKleZum14a,NMBRTHY:GraKleZum14b}.

The heart of this paper relies on the following two observations.
Firstly, to speed up the individual discrete logarithm phase, we start by
  speeding up the initial splitting step, and for that we compute a
  representation of a preimage of the given target of smaller degree, and/or whose
  coefficients are smaller. It will improve its smoothness probability.
Secondly, to compute this preimage of smaller degree, we
exploit the proper subfields of the finite field $\GF(p^n)$, and
intensively use this key-ingredient: since we
  are computing discrete logarithms modulo (a prime divisor of)
  $\Phi_n(p)$, we can freely multiply or divide the target by any
  element in a proper subfield without affecting its discrete
  logarithm modulo $\Phi_n(p)$. 

\subsection*{Organization of the paper.}
The background needed is presented as preliminaries in Section~\ref{prelim}.
We present our generic strategy to lower the degree of the
polynomial representing a given element in $\GF(p^n)$
in Section~\ref{sec:general-strategy}.
We apply it to characteristic two and three in
Section~\ref{sec:app-small-char}.
Preliminaries before the large characteristic case are given in
Section~\ref{sec:preliminaries}. 
We apply our technique
to medium and large characteristic finite fields, that is the NFS case
and its tower variant in Section~\ref{sec:app-large-char}, and provide
examples of cryptographic size in Section~\ref{sec:large-char-examples}. 
Finally in Section~\ref{sec:optimal-repr} we present a more advanced
strategy, to exploit several subfields at a time, and we apply it
to $\F_{p^6}$.

\section{Preliminaries}
\label{prelim}

\subsection{Setting}
In this paper, we are interested in nonprime finite fields
$\GF(p^n)$, $n>1$. To keep the same notation between small, medium, and large
characteristic finite fields, we assume that the
field $\F_{p^n}$ is defined by an extension of degree $n_2$ above an extension
of degree $n_1$, that is, $\F_{(p^{n_1})^{n_2}}$, and $n=n_1n_2$. 
The elements are of the form $T =
\sum_{i=0}^{n_2-1}\sum_{j=0}^{n_1-1}a_{ij}y^jx^i$, where the
coefficients $a_{ij}$ are in $\F_p$, the coefficients $a_i
=\sum_{j=0}^{n_1-1}a_{ij}y^j$ are in
$\F_{p^{n_1}} = \F_p[y]/(h(y))$,
and $\F_{(p^{n_1})^{n_2}} = \F_{p^{n_1}}[x]/(\psi(x))$, where $h$
is a monic irreducible polynomial of $\F_p[y]$ of degree $n_1$ and $\psi$ is a
monic irreducible polynomial of $\F_{p^{n_1}}[x]$ of degree $n_2$. 
In other words, $T$ is represented as a
polynomial of degree $n_2-1$ in the variable $x$, and has coefficients
$a_i \in \F_{p^{n_1}}$.
For the FFS and NFS algorithms,
$n_1=1$ and $n_2=n$; for finite fields from pairing constructions, $n_2
> 1$ is a strict divisor of $n$, and for the original version of TNFS, $n_1=n$ and $n_2=1$.

\begin{definition}[Smoothness]
  Let $B$ be a positive integer.
A polynomial is said to be $B$-smooth w.r.t.~its degree if all its irreducible
factors have a degree smaller than $B$.
An integer is said to be 
$B$-smooth if all its prime divisors are less than $B$. An ideal in a
number field is said to be $B$-smooth if it factors into prime ideals whose
norms are bounded by $B$.
\end{definition}

\begin{definition}[Preimage]
  The preimage of an element
  $a=\sum_{i=0}^{n_2-1}\sum_{j=0}^{n_1-1}a_{ij}y^jx^i \linebreak \in
  \F_{(p^{n_1})^{n_2}}$ will be, for the NFS and TNFS algorithms, the bivariate polynomial
  $\sum_{i=0}^{n_2-1}\sum_{j=0}^{n_1-1}a'_{ij}y^jx^i \in \ZZ[x,y]$,
  where each coefficient $a'_{ij}$ is a lift in $\ZZ$ of the
  coefficient $a_{ij}$ in $\F_p$.
  It is a preimage for the reduction modulo $(p, h,
\psi)$, that we denote by $\rho: \ZZ[x,y] \to \F_{(p^{n_1})^{n_2}}$.
  In small characteristic, the preimage of $a$ is a univariate polynomial in
  $\F_{p^{n_1}}[x]$.   It is a preimage for the reduction modulo
  $\psi$, that we also denote by $\rho: \F_{p^{n_1}}[x] \to \F_{(p^{n_1})^{n_2}}$. 
\end{definition}

\begin{definition}[Pseudonorm]
  The integral pseudonorm w.r.t.~a number field $\QQ[x]/(f(x))$
  ($f$ monic) of a polynomial $T=\sum_{i=0}^{\deg
    f-1}a_ix^i$ of integer coefficients $a_i$ is computed as
  $\Reslt_x(T(x),f(x))$.
\end{definition}

Since there is no chance for a preimage of a target $T_0$ to be $B$-smooth,
the individual discrete logarithm is done in two steps: an
\emph{initial splitting} of the target,\footnote{also called
\emph{boot} or \emph{smoothing step} in large characteristic finite fields} and then 
a \emph{descent} phase.\footnote{in order to make
  no confusion with the mathematical \emph{descent}, which is not involved
in this process, we mention that in this
step, the norm (with NFS) or the degree (with FFS) of the preimage \emph{decreases}.}
The initial splitting is an iterative process that tries many
targets $g^t T_0 \in \F_{p^n}^*$, where $t$ is a known exponent
(taken uniformly at random), until a $B_1$-smooth
decomposition of the preimage is found.
Here \emph{smooth} stands for a
factorization into irreducible polynomials of $\F_{p^{n_1}}[x]$ of
degree at most $B_1$ in the small characteristic setting, resp., a pseudonorm
that factors as an integer into a product of primes smaller than
$B_1$ in the NFS (and TNFS) settings.

The second phase starts a recursive process for each element
less than $B_1$ but greater than $B_0$ obtained after the initial
splitting phase. 
Each of these medium-sized elements are processed until a complete
decomposition over the factor basis is obtained. Each element obtained
from the initial splitting is at the root of its descent tree. 
One finds a relation involving the original one and
other ones whose degree, resp., pseudonorm, is strictly smaller than
the degree, resp., pseudonorm, of the initial element at the root.
These smaller elements form the new leaves of the descent
tree. For each leaf, the process is repeated until all the leaves
are elements in the factor basis. The discrete logarithm of an element
output by the initial splitting can be computed by a tree traversal. 
This strategy is considered in \cite[\S 6]{CopOdlSch86}, 
\cite[\S 7]{LaMOdl91},
\cite[\S 3.5]{JouLer03}, 
\cite[\S 4]{PKC:ComSem06}.

In small characteristic, the initial splitting step is known as
the Waterloo\footnote{the name comes from the
  authors' affiliation: the University of Waterloo, ON, Canada.}
algorithm~\cite{BFMV84,C:BlaMulVan84}.
It outputs $T = U(x)/V(x) \mod
I(x)$, and $U,V$ are two polynomials of degree $\lfloor (n_2-1)/2 \rfloor$.
It uses an Extended GCD computation.
For prime fields, the continued fraction algorithm was already used
with the Quadratic Sieve and Coppersmith-Odlyzko-Schroeppel algorithm.
It expresses an integer $N$
modulo $p$ as a fraction $N \equiv u/v \bmod p$, and the numerator and
denominator are of size about the square root of $p$.
The generalization of this technique was used in \cite{C:JLSV06}.
As for the Waterloo algorithm, this 
technique provides a very good practical speed up but does not improve
the asymptotic complexity.

this subfield tool was highlighted in \cite{AC:Guillevic15}; we will
intensively use it.
\begin{lemma}[{\cite[Lemma 1]{AC:Guillevic15}}] 
\label{lemma:log equality up to subgroup elt}
Let $T \in \FF_{p^n}^{*}$,
and let $\deg T <
n$. Let $\ell $ be a nontrivial prime divisor of $\Phi_n(p)$. 
Let $T' = u \cdot T$ with $u$ in a proper subfield of $\FF_{p^n}$. Then
\begin{equation}
\log T' \equiv \log T \bmod \Phi_n(p) \mbox { and in particular } \log T' \equiv \log T \bmod \ell ~.
\end{equation}
\end{lemma}

\section{The heart of our strategy: representing elements in the cyclotomic subgroup of a
  nonprime finite field with less coefficients}
\label{sec:general-strategy}

In the FFS setting, $n_1 = 1$ and usually $n_2$ is prime and our
technique cannot be helpful, but if $n$ is not
prime, our algorithm applies, and moreover in favorable cases Joux's  $L[1/4]$
algorithm and its variants
can be used and our technique can provide a further notable speed-up in the descent. 
For the implementations in small characteristic, the factor basis is made of the
irreducible polynomials of $\F_{p^{n_1}}[x]$ of very small degree,
e.g., of degrees 1, 2, 3, and 4 in \cite{EPRINT:ACCMOR16}.
Our aim is to improve the smoothness probability of a preimage $P \in
\F_{p^{n_1}}[x]$ of a given target $T \in \F_{(p^{n_1})^{n_2}}$ 
and for that we want to reduce the degree in
$x$ of the preimage $P$
(as a lift of $T$ in $\F_{p^{n_1}}[x]$, $P$ has degree at most $n_2-1$ in $x$), 
while keeping the property
$$ \log(\rho(P)) = \log T \mod \ell~,$$
where $\rho: \F_{p^{n_1}}[x] \to \F_{(p^{n_1})^{n_2}}$ is the
reduction modulo $\psi$.

Let $d$ denote the largest proper divisor of $n$, $1<d<n$ ($d$ might sometimes be equal
to $n_2$ in the QPA setting).
We will compute $P$ in $\F_{p^{n_1}}[x]$ of degree at most $n_2-d/n_1$
in $x$ (and coefficients in $\F_{p^{n_1}}$) such that
\begin{equation}
  \label{eq:P-eq-uT}
  P = u T \pmod{\psi}, \mbox{ where } u^{p^d-1}=1~.
\end{equation}
It means that we will cancel the $d/n_1-1$ higher coefficients (in
$\F_{p^{n_1}}$) of a preimage of $T$ in $\F_{p^{n_1}}[x]$.

There are two strategies: either handle coefficients in $\F_p$ or in
$\F_{p^{\gcd(d,n_1)}}$. We will consider the latter case.
Let $d'=d/\gcd(d,n_1)$ to simplify the notation, and let $[1,U,
\ldots, U^{d'-1}]$ be a polynomial basis of 
$\F_{p^{d'}}$. Every product $P = U^{i}T$ satisfies \eqref{eq:P-eq-uT}.
Define the $d' \times n_2$ matrix $L$ whose rows are made
of the coefficients (in $\F_{p^{n_1}}$) of $U^iT$ for $0\leq i \leq
d'-1$:
$$L_{d' \times n_2} =\begin{bmatrix}
  T \\
  UT \\
  \vdots \\
  U^{d'-1} T
\end{bmatrix} \in \mathcal{M}_{d',n}(\F_{p^{n_1}})~.$$
Then we compute a row-echelon form of this matrix by performing only
$\F_{p^{\gcd(n_1,d)}}$-linear operations over the rows, so that each
row of the echeloned matrix is a $\F_{p^{\gcd(n_1,d)}}$-linear
combination of the initial rows, that can be expressed as 
$$ P = \sum_{i=0}^{d'-1} \lambda_iU^iT = uT,~ \mbox{where }
\lambda_i \in \F_{p^{\gcd(n_1,d)}}~,~ U^i \in \F_{p^{d/\gcd(n_1,d)}}$$
so that $P = u T$ with $u^{p^d-1} = 1$.
Assuming that the matrix is lower-triangular (the other option being
an upper-triangular matrix), we take the first row of the matrix 
 as the coefficients of a polynomial in $\F_{p^{n_1}}[x]$ of degree at
most\footnote{$n_2-d/n_1$ is not necessarily an integer, meaning that the leading
  coefficient of the polynomial is some element in $\F_{p^{n_1}}$. Its
degree in $x$ is actually $n_2-\lceil d/n_1\rceil$.}
$n_2- d/n_1 $. %% 
This is formalized in Algorithm~\ref{alg:lower-degree-poly}.
\setlength{\algomargin}{\leftskip}
\begin{algorithm}
  \DontPrintSemicolon
  \caption{Computing a representation by a polynomial of smaller degree}
  \label{alg:lower-degree-poly}
  \KwIn{Finite field $\F_{p^n}$ represented as a tower
    $\F_{(p^{n_1})^{n_2}} = \F_{p^{n_1}}[x]/(\psi(x))$ (one may have $n_1=1$),
    a proper divisor $d$ of $n$ ($d\mid n$, $1<d < n$),
    $T \in \F_{p^n}$
  }
  \KwOut{$P \in \F_{p^{n_1}}[x]$ a polynomial
    of degree $\leq n_2- d/n_1$
    satisfying $ P \bmod \psi = u T$, where $u \in \F_{p^d}$}
  $d'=d/\gcd(n_1,d)$\;
  Compute a polynomial basis $(1, U, U^2, \ldots, U^{d'-1})$ of the 
    subfield $\F_{p^{d'}}$\;
  Define $L =
  \begin{bmatrix}
    T \\
    U T \\
    \vdots \\
    U^{d'-1} T \\
  \end{bmatrix}$ a $d'\times n_2$ matrix of coefficients in $\F_{p^{n_1}}$\;
  $M \gets$ RowEchelonForm$(L)$ with only $\F_{p^{\gcd(n_1,d)}}$-linear combinations \; 
  $P(x) \gets $ polynomial from the coefficients of the first row of
  $L$ \;
  \Return $P(x)$ \;
\end{algorithm}
We obtain the following Theorem~\ref{th:small-degree-preimage}.
\begin{theorem} \label{th:small-degree-preimage}
  Let $\F_{p^n}$ be a finite field represented
  as a tower $\F_{(p^{n_1})^{n_2}}$.
  Let $T \in \FF_{p^n}^{*}$ be an element which is not in a proper subfield
  of $\FF_{p^n}$.
  Let $d$ be the largest proper
  divisor of $n$, $1 < d < n$ ($n$ is not prime).
  Assume that $T$ is represented by a polynomial in $\F_{p^{n_1}}[x]$ of
  degree larger than $n_2 - d/n_1$.
Then there exists
a preimage $P$ of $T$, in $\F_{p^{n_1}}[x]$, of degree $n_2 - \lceil d/n_1
\rceil$ in $x$ and coefficients in $\F_{p^{n_1}}$, and such that
$$ \log(\rho(P)) = \log T \mod \Phi_n(p)~.$$
\end{theorem}

\begin{proof}
We use Algorithm~\ref{alg:lower-degree-poly} to compute $P$.
The matrix has full rank since the $U^i$s form a polynomial basis of
$\F_{p^{d'}}$. The linear combinations involve $T$ and elements
in $\F_{p^{d/\gcd(n_1,d)}}$ and $\F_{p^{\gcd(n_1,d)}}$ that are in the
proper subfield $\F_{p^d}$ by construction. The first row after
Gaussian elimination will have at least $d/n_1 -1$
coefficients equal to zero at the right, and will represent a polynomial $P$ of degree
at most $n_2 - d/n_1$, that satisfies $P = u T \pmod{\psi}$ where $u = \sum
\lambda_i U^i \in \F_{p^d}$, since in the
process, $T$ was multiplied only
by elements whose images in $\F_{(p^{n_1})^{n_2}}$ are in the subfield
$\F_{p^d}$. We have $\rho(P) = uT$, $u \in \F_{p^{d}}$, and the equality
of logarithms follows by Lemma~\ref{lemma:log equality up to subgroup elt}. 
\end{proof}

We can now directly apply Algorithm~\ref{alg:lower-degree-poly} to
improve the initial splitting algorithm in practice.

\section{Application to small characteristic finite fields, \\
  and cryptographic-size examples}
\label{sec:app-small-char}
In all the examples of small characteristic finite fields from pairings, $n$ is
not prime, for instance $n=6\cdot 509$.
The notation in \cite{PAIRING:AMOR13} was $n=l k$, with the property $p^l
\approx k$. With our notation, $n_1=l$ and $n_2 = k$. 

\subsection{Algorithm}
We directly use Algorithm~\ref{alg:lower-degree-poly} as a subroutine of
Algorithm~\ref{alg:boot in small char}. Then to improve it in
practice, we list valuable modifications.

\begin{algorithm}
  \DontPrintSemicolon
  \caption{Initial splitting in small characteristic with the subfield technique}
  \label{alg:boot in small char}
  \KwIn{Finite field $\F_{p^n}$ of small characteristic
    (e.g., $p=2,3$), with a tower representation $ \F_{p^n}
    =\F_{(p^{n_1})^{n_2}} = \F_{p^{n_1}}[x]/(I(x))$ (one may have $n_1=1$),
    generator $g$ (of the order $\ell$ subgroup of the cyclotomic
    subgroup of $\F_{p^n}$),
    target $T_0 \in \F_{(p^{n_1})^{n_2}}$,
    smoothness bound $B_1$}
  \KwOut{$t$, $P \in \F_{p^{n_1}}[x]$ a polynomial
    of degree $\leq n_2-d/n_1$
    such that $\vlog_g \rho(P) = t+\vlog_g T_0 \mod \ell$,
    and $P(x)$ is $B_1$-smooth (w.r.t. its degree in $x$)}
  $d \gets$ the largest divisor of $n$, $1 < d < n$ \;
  $d' \gets d/\gcd(d,n_1)$\;
  Compute $U(x) \in \F_{(p^{n_1})^{n_2}}$ s.t. $(1, U, U^2, \ldots, U^{d'-1})$ is a
  polynomial basis of the subfield $\F_{p^{d'}}$\;
  \Repeat{$P(x)$ is $B_1$-smooth}{
    take $t \in \{1, \ldots, \ell-1\}$ at random \;
    $T \gets g^t T_0$ in $\F_{(p^{n_1})^{n_2}}$ \;
    Define $L =
    \begin{bmatrix}
      T \\
      U T \\
      \vdots \\
      U^{d'-1} T \\
    \end{bmatrix}$ a $d'\times n_2$ matrix of coefficients in $\F_{p^{n_1}}$\;
    $M \gets$ RowEchelonForm$(L)$  (with $\F_{p^{\gcd(d,n_1)}}$-linear Gaussian elimination)\;
    $P(x) \gets $ the polynomial of lowest degree made of the first row of
    $L$ \;
  }
  \Return $t, P(x)$ \;
\end{algorithm}

\begin{remark} \label{rem:double-side-row-echelon}
As was pointed out to us by
F.~Rodr\'{i}guez-Henr\'{i}quez~\cite{Rodriguez15,EPRINT:ACCMOR16}, the elements of 
the form $x^i R(x)$ where $R$ itself is of degree $\leq n_2-d/n_1$ 
are evenly interesting because the discrete logarithm of
$x^i$ can be deduced from the discrete logarithm of $x$, which is
known after linear algebra.

So we can increase the number of elements tested for $B_1$-smoothness
for each exponent $t$ by a factor $d'$ almost for free in the
following way. We again run a
Gaussian elimination algorithm on the matrix $M$ but in the reverse
side, for instance from row one to row $d'$ and left to right if it was done from row
$d'$ to row one and right to left the first time.  
The matrix is in row-echelon form on the left-hand side and on the
right-hand side (the upper right and lower left corners are filled
with zeros). We obtain a matrix $N$ of the form
$$ N =
\begin{bmatrix}
  *      & \ldots & *     & *     & 0      & \ldots & 0 \\
  0      & \ddots &       &       & \ddots & \ddots & \vdots \\
  \vdots & \ddots & \ddots&       &        & \ddots & 0 \\
  0      & \ldots & 0     & *     & \ldots &   *    & * \\
\end{bmatrix}~.$$
The $i$-th row represents a polynomial $P_i' = x^{e_i}P_i$,
where $P_i$ is of degree at most $n_2-d/n_1$, and $e_i \approx (i-1)\gcd(n_1,d)/n_1$.
Since $x$ is in the factor basis (by construction, like all the degree
one polynomials), its logarithm is 
known at this point (after the relation collection and linear algebra
steps), hence the logarithm of any power $x^{e_i}$ is known. It remains to
compute the discrete logarithm of $P_i$.
\end{remark}

In practice there are some technicalities: in the second Gaussian
elimination, if the leading coefficient is zero, then two rows are
swapped, and it cancels the previous Gaussian elimination (computed
at the other end of the matrix) for that row. 
We end up with a matrix which is in row-echelon form on the right and
almost row-echelon form on the left (or vice-versa).
Since each set of subsequent $n_1/\gcd(n_1,d)$ rows produces
polynomials of the same degree, swapping two rows from the same set
will not
change the degree in $x$ of the polynomial.
In average (this is what we observed in our experiments for
$\F_{3^{6\cdot 509}}$ and $\F_{3^{5\cdot 479}}$), 
some rare polynomials will have a degree in $x$ increased by one or two.
This second Gaussian elimination increases the number of tests by a
factor $d'$ at a very cheap cost, since in fact it allows one to share
the cost of computing the $U^iT$ and the two Gaussian eliminations
over $d'$ tests.

\begin{remark}
  If $\gcd(d,n_1)>1$ we can increase the number of rows by a small factor. 
  We perform linear combinations of $n_1/\gcd(d,n_1)$ subsequent rows (all
  giving a polynomial of same degree):
  $\sum_{0\leq j \leq n_1/\gcd(d,n_1)} \mu_{j}r_j$ where $\mu \in \F_{p^{\gcd(d,n_1)}}$,
  and it will result in new rows and new polynomials of same degree.
\end{remark}

\begin{remark}
\label{rem:linear-combinations-rows}
Other improvements are possible \cite{Rodriguez15,EPRINT:ACCMOR16}, for instance
computing \linebreak {$\F_{p^{\gcd(n_1,d)}}$-linear} combinations over a small
number of rows corresponding to polynomials of almost the same 
degree. The resulting polynomial will have degree increased by one
or two, which does not significantly affect its $B_1$-smoothness
probability in practice for cryptographic sizes. This technique
allows one to produce many more candidates, at a very cheap cost of
linear operations in $\F_{p^{n_1}}[x]$.  
\end{remark}

\subsection{Complexity analysis}

\subsubsection{Cost of computing one preimage $P \in \F_{p^{n_1}}[x]$ in the initial
  splitting step.}
We use the notation of Algorithm~\ref{alg:boot in small char}: let $d$ be the largest
proper divisor of $n$ ($d\mid n$, $1<d<n$), and let $d'=d/\gcd(d,n_1)$.
Since $d' \mid d \mid n=n_1n_2$ and $\gcd(d',n_1)=1$,
then $d' \mid n_2$ and $d' \leq n_2$.
The computation of all the $U^iT$ of the matrix $L$ costs at most $d'n_2^2$
multiplications in $\F_{p^{n_1}}$, since a schoolbook multiplication
in $\F_{(p^{n_1})^{n_2}}$ costs $n_2^2$ multiplications in
$\F_{p^{n_1}}$. There are $d'$ such multiplications. 
The complexity of a reduced row-echelon form computation of a $(d'\times
n_2)$-matrix, $d' \leq n_2$, is less than $O(d'^2n_2)$
multiplications in $\F_{p^{n_1}}$ \cite[\S 13.4.2]{HFF:DumPer13}.
To simplify, we consider that the computation of the matrix $L$ and of two
Gaussian eliminations is done in time at most $O(d'n_2^2)$.
This cost is shared over $d'$ polynomials $P_i$ to be tested for $B_1$-smoothness.
In this way, the complexity of computing a preimage $P$ with our
initial splitting algorithm is the same as in the Waterloo algorithm: $O(n_2^2)$,
and moreover the smoothness probabilities are much higher for the
targeted cryptographic cases coming from supersingular pairing-friendly curves.
We also replace two $B_1$-smoothness tests by only one, and that might save
some time in practice (this saving disappears in the $O$ notation). 
We present the theoretical costs in
Tables~\ref{tab:cost-splitting} and~\ref{tab:cost-initial-splitting} from \cite{FlaGouPan01}. 
XGCD stands for extended Euclidean algorithm, SQF stands for
SQuare-free Factorization, DDF stands for Distinct Degree Factorization,
and EDF stands for Equal Degree Factorization.
All the polynomials to be factored are of degree smaller than $n_2$;
we take $n_2$ as an upper bound to get the costs of Table~\ref{tab:cost-initial-splitting}.

\begin{table}[htbp]
  \caption{Costs for the initial splitting step. The
    preimage obtained with Algorithm~\ref{alg:boot in small char} has
    degree $d_P \leq n_2 - d/n_1$.
    The Waterloo algorithm~\cite{BFMV84,C:BlaMulVan84}
    produces two polynomials of degree $d_P=\lfloor n_2/2 \rfloor$.} 
  \label{tab:cost-splitting}
  \centering
  \begin{tabular}{|l|c|}
    \hline Factorization & cost     \\
    \hline
Square-free (SQF)     & $O(d_P^2)$             \\
Distinct degree (DDF) & $O(d_P^3 \log p^{n_1})$ \\
Equal degree (EDF)    & $O(d_P^2 \log p^{n_1})$ \\
\hline
  \end{tabular}
\end{table}

\begin{table}[ht]
  \caption[tab:cost-per-test]{Cost in multiplications in $\F_{p^{n_1}}$ to
    compute one preimage to be tested for smoothness, in the initial splitting
    step.}
  \label{tab:cost-initial-splitting}
  \centering
  \begin{tabular}{|c|c|c|c|}
    \hline Computation & XGCD$(T,I)$ &  \multicolumn{2}{c|}{
  matrix $[U^iT]_{0\leq i \leq d'-1}$ and row echelon form}\\
\hline Algorithm &  Waterloo & this work, Alg.~\ref{alg:boot in small char} &
                      this work + Rem.~\ref{rem:double-side-row-echelon} \\
    \hline Cost & $O(n_2^2)$ & $O(d'n_2^2)$ & $O(n_2^2)$ \\
    \hline
  \end{tabular}
\end{table}

\subsubsection{running time of the initial splitting step}
To start, we recall some results on the smoothness probability of a
polynomial of given degree.
\begin{definition}
  Let $N_q(b; d)$ denote the number of monic polynomials over $\F_q$ of degree $d$ which are
  $b$-smooth.
Let $N_q(b; d_1, d_2)$
denote the number of coprime pairs of monic polynomials over $\F_q$ of
degrees $d_1$ and $d_2$, respectively, which are $b$-smooth.

Let $\Prb_q(b;d)$ denote the probability of a monic polynomial over
$\F_q$ of degree $d$ to be $b$-smooth.
Let $\Prb_q(b;d_1,d_2)$ denote the probability of two coprime monic polynomials over
$\F_q$ of degrees $d_1$ and $d_2$ to be both $b$-smooth.
\end{definition}

Odlyzko gave the following estimation for $\Prb_q(b;d)$ in \cite[(4.5), p.~14]{EC:Odlyzko84}.
\begin{equation}
  \label{eq:prob poly b-smooth}
  \Prb_{q}(b,d)^{-1} = \exp\left((1+o(1)) \frac{d}{b} \log_e \frac{d}{b}\right) \mbox{
  for } d^{1/100} \leq b \leq d^{99/100}~.
\end{equation}

Writing the smoothness bound degree $b = \log L_Q[\alpha_b, c_b]/\log
p^{n_1} $ to match \linebreak Odlyzko's convention $b= c_b n_2^{\alpha_b}(\log n_2)^{1-\alpha_b}$, and
the degree of the polynomial to be tested for smoothness
$d=an_2$, where $a \in ]0,1[$ and $n_2=\log Q/\log p^{n_1}$, one
obtains 
$$
\Prb_{p^{n_1}}(b,d) = L_Q\left[1-\alpha_b,
  -(1-\alpha_b)a/\gamma\right]~, \mbox{ where } Q = p^{n_1n_2}.
$$ 

\begin{theorem}[{\cite[Theorem 1]{DCC:DrmPan02}}]
\label{th:DrmPan02}
  Let $\delta > 0$ be given. Then we have, uniformly for $b,d_1,d_2
  \to \infty $ with $d_1^\delta \leq b \leq d_1^{1-\delta}$ and 
$d_2^\delta \leq b \leq d_2^{1-\delta}$,
$$N_q(b; d_1, d_2) \sim \left( 1 - \frac{1}{q} \right) N_q(b;d_1)
N_q(b; d_2) ~.$$
\end{theorem}

\begin{corollary}[{\cite[Theorem 1]{DCC:DrmPan02}}]
  Let $\delta > 0$ be given. Then we have, uniformly for $b,d_1,d_2
  \to \infty $ with $d_1^\delta \leq b \leq d_1^{1-\delta}$ and 
$d_2^\delta \leq b \leq d_2^{1-\delta}$,
$$ \Prb_q(b; d_1,d_2) \sim \left( 1 - \frac{1}{q} \right) \Prb_q(b; d_1)
\Prb_q(b; d_2)~.$$
\end{corollary}

We can now compare the Waterloo algorithm with this work.
Assuming that $B_1 = \log_{p^{n_1}} L_{p^n}[2/3, \gamma]$ for a
certain $\gamma$, then the 
probability of a polynomial of degree $an_2$, $0<a<n_2$, to be
$B_1$-smooth is $L_{p^n}[1/3, -a/(3\gamma)]$.
In the Water\-loo algorithm, two polynomials of degree $n_2/2$ should be
$B_1$-smooth at the same time, and the expected running time to find such a
pair is $L_{p^n}[1/3, 1/(3\gamma)]$ (the square of $L_{p^n}[1/3, 1(6\gamma)]$).
In our algorithm, a polynomial of degree $\lfloor n_2-d/n_1\rfloor =
\lfloor n_2(1-d/n) \rfloor$ is tested
for $B_1$-smoothness, so finding a good one requires
\begin{equation}
L_{p^n}[1/3, a/(3\gamma)]\mbox{ tests, where }a \approx 1-d/n~,  
\end{equation} 
which is always faster than the Waterloo algorithm, for which
$a=1$. When $n$ is even (this is always the case for 
finite fields of supersingular pairing-friendly curves), one can
choose $d=n/2$, hence $a=1/2$ and our algorithm has running time the
square root of the running time of the Waterloo algorithm.

\subsection{Improving the record computation in GF$({3^{6 \cdot 509}})$}
\label{subsec:3^6*509}
Adj, Menezes, Oliveira, and Rodr{\'i}guez-Henr{\'i}quez estimated in
{\cite{PAIRING:AMOR13}} the cost to compute discrete logarithms in the 
4841-bit finite field GF$(3^{6\cdot 509})$ and announced their
record computation in July 2016 \cite{NMBRTHY:ACCMORR16}. 
The details of the computations are available in Adj's PhD thesis
\cite{PhD:Adj16} and the details for initial splitting and descent can be
found in \cite{Master:CanalesMartinez15}. 
The elements are represented by polynomials of degree at most 508
whose coefficients are in $\F_{3^6}$. In this case $n_1=6$ and $n_2=509$.
The initial splitting made with the Waterloo algorithm outputs two
polynomials of degree 254. The probability that two independent and
relatively prime polynomials of degree 254 over $\F_{3^6}$ are
simultaneously $b$-smooth is $(1-1/3^6) \Pr_{{3^6}}^2(254,b)$
\cite{DCC:DrmPan02}.
The term $(1-1/3^6)$ is negligible in practice for the values that we
are considering.

\subsubsection{Improvements}
Our Algorithm~\ref{alg:boot in small char} outputs \emph{one}
polynomial of degree 254, whose probability to be $b$-smooth is
$\Pr_{{3^6}}(n,b)$, i.e., the square root of the previous one. So
we can take a much smaller $b$ while reaching the same probability as
before with the Waterloo algorithm. We list in 
Table~\ref{tab:p2(n,b) and p(n,b)}, p.~\pageref{tab:p2(n,b) and p(n,b)},
the values of $b$ to obtain a probability between $2^{-40}$ and $2^{-20}$. 
For instance, if we allow $2^{30}$ trials, then we can set $b=28$ with our
algorithm, instead of $b= 43$ previously:
 we have $\Pr_{{3^6}}^2(254,43) = 2^{-30.1}$, and we only need to
take $b=28$ to get the same probability with this work: $\Pr_{{3^6}}(254,28) =
2^{-29.6}$. 
This will provide a good practical speed-up of the descent phase: much
fewer elements need to be ``reduced'': this reduces the initial width
of the tree, and they are of much smaller
degree: this reduces the depth of the descent tree.

\subsubsection{A 30-smooth initial splitting}
The finite field is represented with $n_1=6$ and $n_2=509$, that is, as a first extension
$\F_{3^6} = \F_{p^{n_1}} = \F_{3}[y]/(y^6 + 2 y^4 + y^2 + 2 y + 2)$, then a second
extension $\F_{3^{6\cdot 509}} = \F_{3^6}[x]/(I(x))$, where $I(x)$ is
the degree 509 irreducible factor of $h_1x^{q_1} - h_0$, where $q_1=p^{n_1}$, $h_1 = x^2
+ y^{424} x$, and $h_0 = y^{316} x + y^{135}$. 
The generator is $g = x + y^2$. 
As a proof of concept, we computed a 30-smooth initial splitting of
the target
$T_0 = \sum_{i=0}^{508} (y^{\lfloor \pi (3^6)^{i+1} \rfloor} \bmod
3^6) x^i$, with the parameters $d=3\times 509$, $d'=d/\gcd(d,n_1) =
509$. Each trial $g^t T_0$ produces $d'=509$ polynomials to test for
smoothness.
We found that  $g^{47233} T_0 = u v x^{230} P$, where $u=1 \in
\F_{3^6}$, $v \in \F_{3^{3\cdot 509}}$, and $P$ is of degree
255 and 30-smooth. The equality $(g^{47233} T_0)^{\frac{p^n-1}{\ell}} =
(u v x^{230} P)^{\frac{p^n-1}{\ell}}$ is satisfied. 
 The explicit value of $P$ is available at
 \url{https://members.loria.fr/AGuillevic/files/F3_6_509_30smooth.mag.txt}.

The whole computation took less than 6 days (real time) on 48 cores Intel Xeon
E5-2609 at 2.40GHz (274 core days, i.e., 0.75 core-years).
This is obviously an overshot compared to the estimate of $2^{26.6}$, but
this was done with a nonoptimized Magma implementation.

As a comparison, with the classical Waterloo algorithm, Adj et al.~
computed a 40-smooth initial splitting in 51.71 CPU (at 2.87GHz) years
\cite[Table~5.2, p.~87]{PhD:Adj16} and \cite{NMBRTHY:ACCMORR16}.
They obtained irreducible polynomials of degree 40, 40, 39, 38, 37,
and seven polynomials of degree between 22 and 35.
They needed another 9.99 CPU years (at 2.66 GHz) to compute a
classical descent from 40-smooth to 21-smooth polynomials.
A complete comparison can be found in \cite{EPRINT:ACCMOR16} and \cite{ECC:Adj16}.
In \cite{EPRINT:ACCMOR16}, Adj et al.~estimated that with our
Algorithm~\ref{alg:boot in small char} enriched as in Remarks~\ref{rem:double-side-row-echelon}
and~\ref{rem:linear-combinations-rows}, it is possible to compute 
discrete logarithms in $\F_{3^{6\cdot 709}}$ at the same cost as in 
$\F_{3^{6\cdot 509}}$ with the former Waterloo algorithm.

\begin{table}[hbtp]
\caption{Smoothness probabilities of polynomials over finite fields,
  comparison of the Waterloo algorithm and Algorithm~\ref{alg:boot in small char}.
The values were computed with Odlyzko's induction
formula~\cite{EC:Odlyzko84} and Drmota and Panario's
Theorem~\ref{th:DrmPan02}, as in \cite{Master:CanalesMartinez15}.} 
\label{tab:smoothness-probas}
  \centering
\subfloat[{Probabilities for $\GF(3^{6\cdot 509})$}]
{
  \label{tab:p2(n,b) and p(n,b)}
  \begin{tabular}{|c|c||c|c|}
    \hline \multicolumn{2}{|c||}{Waterloo alg.} & \multicolumn{2}{c|}{Algorithm~\ref{alg:boot in small char}} \\
    \hline $b$ & $\Pr_{3^6}^2(254,b)$ & $b$ & $\Pr_{3^6}(254, b)$ \\
    \hline  
     36 & $2^{-40.1}$ & 22 & $2^{-42.3}$ \\
     37 & $2^{-38.4}$ & 23 & $2^{-39.6}$ \\
     38 & $2^{-36.8}$ & 24 & $2^{-37.2}$ \\
     39 & $2^{-35.3}$ & 25 & $2^{-35.1}$ \\
     40 & $2^{-33.9}$ &    &           \\
     41 & $2^{-32.5}$ & 26 & $2^{-33.1}$ \\
     42 & $2^{-31.3}$ & 27 & $2^{-31.3}$ \\
     43 & $2^{-30.1}$ & 28 & $2^{-29.6}$ \\
     44 & $2^{-28.9}$ &    &           \\
     45 & $2^{-27.9}$ & 29 & $2^{-28.1}$ \\
     46 & $2^{-26.9}$ &    &           \\
     47 & $2^{-25.9}$ & 30 & $2^{-26.6}$ \\
     48 & $2^{-25.0}$ & 31 & $2^{-25.3}$ \\
     49 & $2^{-24.1}$ & 32 & $2^{-24.1}$ \\
     50 & $2^{-23.3}$ & 33 & $2^{-23.0}$ \\
     51 & $2^{-22.5}$ &    &           \\
     52 & $2^{-21.8}$ & 34 & $2^{-21.9}$ \\
     53 & $2^{-21.1}$ & 35 & $2^{-21.0}$ \\
     54 & $2^{-20.4}$ & 36 & $2^{-20.1}$ \\
     55 & $2^{-19.7}$ &    &           \\
     56 & $2^{-19.1}$ & 37 & $2^{-19.2}$ \\
     57 & $2^{-18.5}$ &    &           \\
     58 & $2^{-18.0}$ & 38 & $2^{-18.4}$ \\
     59 & $2^{-17.4}$ & 39 & $2^{-17.6}$ \\
     60 & $2^{-16.9}$ & 40 & $2^{-16.9}$ \\
     61 & $2^{-16.4}$ & 41 & $2^{-16.3}$ \\
     62 & $2^{-15.9}$ & 42 & $2^{-15.6}$ \\
    \hline  
  \end{tabular}
}
\subfloat[{For $\GF(3^{5\cdot 479})$}]
{
  \label{tab:smoothness-proba-GF-3-5-479}
  \begin{tabular}{|c|c|c|}
    \hline & Waterloo alg. & Algorithm~\ref{alg:boot in small char} \\
    \hline $b$ & $\Pr_{3^5}^2(239,b)$ & $\Pr_{3^5}(383,b)$ \\
    \hline
24 & $2^{-67.96}$ & $2^{-67.59}$ \\
25 & $2^{-63.95}$ & $2^{-63.86}$ \\
26 & $2^{-60.30}$ & $2^{-60.45}$ \\
27 & $2^{-56.95}$ & $2^{-57.32}$ \\
28 & $2^{-53.89}$ & $2^{-54.44}$ \\
29 & $2^{-51.07}$ & $2^{-51.79}$ \\
30 & $2^{-48.46}$ & $2^{-49.34}$ \\
31 & $2^{-46.06}$ & $2^{-47.06}$ \\
32 & $2^{-43.83}$ & $2^{-44.95}$ \\
33 & $2^{-41.76}$ & $2^{-42.99}$ \\
34 & $2^{-39.83}$ & $2^{-41.16}$ \\
35 & $2^{-38.03}$ & $2^{-39.44}$ \\
36 & $2^{-36.35}$ & $2^{-37.84}$ \\
37 & $2^{-34.77}$ & $2^{-36.34}$ \\
38 & $2^{-33.30}$ & $2^{-34.92}$ \\
39 & $2^{-31.91}$ & $2^{-33.60}$ \\
40 & $2^{-30.61}$ & $2^{-32.34}$ \\
41 & $2^{-29.39}$ & $2^{-31.16}$ \\
42 & $2^{-28.23}$ & $2^{-30.05}$ \\
43 & $2^{-27.14}$ & $2^{-28.99}$ \\
44 & $2^{-26.11}$ & $2^{-27.99}$ \\
45 & $2^{-25.13}$ & $2^{-27.04}$ \\
46 & $2^{-24.21}$ & $2^{-26.14}$ \\
47 & $2^{-23.33}$ & $2^{-25.29}$ \\
48 & $2^{-22.50}$ & $2^{-24.47}$ \\
49 & $2^{-21.71}$ & $2^{-23.70}$ \\
50 & $2^{-20.96}$ & $2^{-22.96}$ \\
\hline
  \end{tabular}
}
\end{table}

\subsection{Computing discrete logarithms in $\F_{2^{512}}$ and $\F_{2^{1024}}$}

In \cite[\S 3.6]{EC:GJMN16} discrete logarithms in 
$\F_{2^{512}}$ and $\F_{2^{1024}}$ need to be computed modulo the full
multiplicative group order $2^n-1$. 
As pointed to us by R.~Granger, our technique can be used  to compute
discrete logarithms in $\F_{2^{1024}}$. Our algorithm provides a
decomposition of the target as the product $u R$ where $u$ is an element in the
largest proper subfield $\F_{2^{512}}$, and
$P$ is an element of $\F_{2^{1024}}$ of degree 512 instead of 1023. 
The discrete logarithm of the
subfield cofactor $u$ can be obtained by a discrete logarithm
computation in $\F_{2^{512}}$. 
More generally, our technique is useful when discrete logarithms in
nested finite fields such as $\F_{2^{2^i}}$ are computed recursively.

\subsection{Improving the record computation in GF$({3^{5\cdot 479}})$}
\label{ex:JouxPierrot-GF-3-5-479}

Joux and Pierrot announced a discrete logarithm record computation in
GF$(3^{5\cdot 479})$ in \cite{NMBRTHY:JouPie14} (then published in
\cite{AC:JouPie14}). They defined a first degree 5 extension $\F_{3^5} 
= \F_3[y]/(y^5 - y + 1)$ and then a degree 479 extension on top of $\F_{3^5}$. With our
notation, we have $p=3$, $n_1=5$, and $n_2=479$. 
The irreducible degree 479 polynomial $I(x)$ is a divisor of 
$x h_1(x^{q_1}) - h_0 (x^{q_1})$, where $q_1=p^{n_1}=3^5$, $h_0 = x^2 + y^{111} x$ and $h_1
= y x + 1$.
Given a target $T \in \F_{3^{5\cdot 479}}$, the Waterloo initial
splitting outputs two polynomials $u(x), v(x) \in \F_{3^5}[x]$ of
degree $\lfloor 478/2 \rfloor = 239$.
Our Algorithm~\ref{alg:lower-degree-poly}
outputs one polynomial of degree $\lfloor \frac{4}{5} 479
\rfloor = 383$. This example is interesting because the smoothness
probabilities are very close. We computed the exact values with
Drmota--Panario's formulas, and give them in
Table~\ref{tab:smoothness-proba-GF-3-5-479},
p.~\pageref{tab:smoothness-proba-GF-3-5-479}.
We obtain $\Pr_{{3^5}}^2(239,50) = 2^{-20.96}$ (Waterloo) and $\Pr_{{3^5}}(383,50) =
2^{-22.96}$, i.e., our
Algorithm~\ref{alg:boot in small char} would be four times slower
compared to Joux's and Pierrot's record; 
 $\Pr_{{3^5}}^2(239,40) = 2^{-30.61}$ and $\Pr_{{3^5}}(383,40) =
2^{-32.34}$; 
  $\Pr_{{3^5}}^2(239,30) = 2^{-48.46}$ and $\Pr_{{3^5}}(383,30) = 2^{-49.34}$~; and the
  cross-over point is for $b=24$: in this case, we have 
 $\Pr_{{3^5}}^2(239,24) = 2^{-67.96}$ and $\Pr_{{3^5}}(383,24) = 2^{-67.59}$, which is
 slightly larger. 
 
The probabilities would advise using the classical initial splitting
with the Waterloo (extended GCD) algorithm. We remark that this algorithm
would output two $B_1$-smooth polynomials of degree $(n_2-1)/2$.
Each would factor into at least $(n_2-1)/(2B_1)$ irreducible polynomials of
degree at most $B_1$. Each such factor is sent as an input to the
second step (descent step), that is, roughly $n_2B_1$ factors. 
If we use Algorithm~\ref{alg:boot in small char}, the initial splitting will
outputs one polynomial of degree $4/5n_2= 383$ that factors into at least
$4/5 n_2/B_1$ polynomials of
degree at most $B_1$, each of them sent as input to the second step,
that is, the descent step is called 20\% time less, 
and that would reduce the total width of the descent tree
in the same proportion. Since the descent is the most costly part, and in
particular, the memory size required is huge, this remark would need
to be taken into consideration for a practical implementation.

As a proof of concept of our algorithm, we implemented in Magma our
algorithm, took the same parameters, generator, and target as in
\cite{NMBRTHY:JouPie14}, and found a 50-smooth decomposition for the
target given by the 471-th row of the matrix computed for $g^{23940}
T_0$ in 1239 core-hours (22.12 hours over 56 cores) on an Intel Xeon E5-2609 at 2.40GHz
 (compared to 5000 core-hours announced in
 \cite{NMBRTHY:JouPie14}).

 The value can be found at
\url{https://members.loria.fr/AGuillevic/files/F3_5_479_50smooth.mag.txt}.
In our technique, we compute $g^t T_0 = u v R$ where $u \in \F_{3^5}$
(this is the leading term of the polynomial), $v \in \F_{3^{479}}$, and
$R$ is 50-smooth. The discrete logarithm of $u$ can be
tabulated, however it remains quite hard to compute the discrete
logarithm of $v$. Our technique is useful if it is  easy (or not
required) to compute discrete logarithms in the subfields.

\section{Preliminaries before medium and large characteristic cases}
\label{sec:preliminaries}

In the first part of this paper, we were considering polynomials, and
we wanted polynomials of smallest possible degree. Now we turn to the
medium and large characteristic cases, where we do not have
polynomials but ideals in number fields, and we want ideals of small
norm. It requires testing whether large integers (norms) are smooth as fast as possible.
We recall the results of Pomerance and Barbulescu on the
\emph{early abort strategy}.
\subsection{Pomerance's Early Abort Strategy}
\label{sec:pomerance}

Pomerance in \cite{Pomerance82} introduced the
\emph{Early Abort Strategy} (EAS) 
to speed up the factorization of large integers, within Dixon's
algorithm, the Morrison--Brillhart (continued fraction) algorithm,
and the Schroeppel (linear sieve) and quadratic sieve, with several variations
in the factorization sub-routine (trial-division, Pollard--Strassen method). 
The Early Abort Strategy 
provides an asymptotic improvement in the
expected running time. Two versions are studied in \cite{Pomerance82}:
one early-abort test, then many tests.
In the relation collection step of the NFS algorithm, the partial
factorization of the pseudonorms
 is done with ECM in time $L_Q[1/6]$ ($Q=p^n$),
which is negligible compared to the total cost in $L_Q[1/3]$. So
Pomerance's EAS does not provide an asymptotic
speed-up, but a practical one.
However, in the individual discrete logarithm computation, the
initial splitting requires to find smooth integers (pseudonorms)
of larger size: $L_Q[1]$. This time the ECM cost is not
negligible, and Pomerance's EAS matters. The speed-up was analyzed by
Barbulescu in \cite{PhD:Barbulescu13}.

\begin{remark}
  Instead of the ECM test, it could be possible to use the hyperelliptic curve method
  test of H.~Lenstra, Pila and Pomerance~\cite{LenPilPom93,LenPilPom02}.
  This was investigated for instance by Cosset \cite[Chapter~4]{PhD:Cosset11}.
\end{remark}

Pomerance's analysis is presented in the
general framework of testing integers for smoothness. This is named
\emph{smoothing problem} in \cite[Chapter~4]{PhD:Barbulescu13}. 
In the individual discrete logarithm context, the numbers we want to
test for smoothness are not integers in an interval,
but pseudonorms, and their chances of being smooth do not exactly match
the chances of random integers of the same size.
However, we will make the usual heuristic assumption that for our
asymptotic computations, the pseudonorms considered behave as
integers of the same size.
We give Pomerance's Early Abort Strategy with one test in
Algorithm~\ref{alg:EAS} and with $k$ tests in Algorithm~\ref{alg:k-EAS}.

\begin{algorithm}
  \DontPrintSemicolon
   \caption{Pomerance's Early Abort Strategy (EAS)}
   \label{alg:EAS}
   \KwIn{Integer $m$, smoothness
     bound $B_1$, real numbers $\theta, b \in ]0, 1[$
}
   \KwOut{$B_1$-smooth decomposition of $m$, or $\bot$}
   $(m_0,m_1) \gets$ \texttt{ECM} $(m, B_1^\theta)$ 
   \tcp*[r]{cost: $L_{B_1^\theta}[1/2,\sqrt{2}]$}
   \tcp*[r]{$m_0$ is a $B_1^\theta$-smooth part of $m$}
   \tcp*[r]{$m_1$ is the non-factorized part of $m$}
   \If{$m_1 \leq m^{1-b}$}{
     $(m_2,m_3) \gets$ \texttt{ECM}$(m_1, B_1)$ \tcp*[r]{cost: $L_{B_1}[1/2,\sqrt{2}]$}
     \If{$m_3 = 1$}{
       \Return{$B_1$-smooth decomposition $m_1,m_2$ of $m$}
     }
   }
   \Return $\bot$  
\end{algorithm}

\begin{algorithm}[hbtp]
  \DontPrintSemicolon
   \caption{Pomerance's Early Abort Strategy with $k$ tests ($k$-EAS)}
   \label{alg:k-EAS}
   \KwIn{Integer $m$,
     smoothness bound $B_1$, number of tests $k \geq 0$, \\
     array of positive real numbers 
     $\mathbf{b} = [b_0, b_1, \ldots, b_k]$ where $0 <b_i \leq 1$, and
     $\sum_{i=0}^k b_i = 1$ \\
     array of positive real numbers 
     $\boldsymbol{\theta} =[\theta_0, \ldots, \theta_{k} = 1]$ where
     $\theta_i < \theta_{i+1}$ \\}
   \KwOut{$B_1$-smooth decomposition of $m$, or $\bot$}
     $m_i \gets m$\;
     $i \gets 0$ \; 
     $S \gets \emptyset $ \;
     \Repeat{$(i > k)$ \texttt{OR} $(m_i = 1)$ \texttt{OR} $(m_i > m^{1-\sum_{j=0}^{i-1} b_j})$}{
       $(s_{i},m_{i+1}) \gets$ \texttt{ECM} $(m_i, B_1^{\theta_i})$ 
       \tcp*[r]{cost: $L_{B_1^{\theta_i}}[1/2,\sqrt{2}]$}
       \tcp*[r]{$s_{i}$ is a $B_1^{\theta_i}$-smooth part of $m_i$,
         $m_{i+1}$ is not factorized}
       $S \gets S \cup s_i$ \;
       $m_i \gets m_{i+1}$ \;
       $i \gets i+1$ \;
     }
  \If{$m_i == 1$}{
    \Return{$B_1$-smooth decomposition $S$ of $m$}
  }
  \Return $\bot$
 \end{algorithm}
 
Writing the complexities as in Pomerance's paper, in terms of $k$
early-abort tests, one obtains Theorems~\ref{th:EAS-cpx} and~\ref{th:k-EAS-cpx}.

\begin{theorem}[{\cite[\S~4.3]{PhD:Barbulescu13}}]\label{th:EAS-cpx}
     The expected running time of the smoothing problem of an
   integer $N$ with Pomerance's EAS and the ECM smoothness test 
   is $L_{N}[1/3, c]$ where $c = (23/3)^{2/3}/3$,
   the smoothness bound is $B = L_N[2/3, \gamma]$,
   where $\gamma = 1/c$,
   $\theta = 4/9$, and $b=8/23$. 
\end{theorem}

\begin{theorem}[{\cite[\S~4.5 Th.~4.5.1]{PhD:Barbulescu13}}]\label{th:k-EAS-cpx}
The expected running time of the smoothing problem of an integer
$N$ with $k$ tests of Pomerance's EAS and the ECM smoothness test 
is $L_N[1/3, c]$ where
$$ c = 3^{1/3} ((15+4(2/3)^{3k})/19)^{2/3}~,
$$ 
the smoothness bound is $B = L_N[2/3, \gamma]$, where
$$  \gamma = 1/c ~,$$
the bound $b_i$ for $0 \leq i \leq k-1$ on the remaining part $m_i$
in Algorithm~\ref{alg:k-EAS} is 
$$ b_i = (2/3)^{3(k-i)} 19 / (15 + 4 (2/3)^{3k}) ~, $$
and the exponent $\theta_i$ for $0 \leq i \leq k$ is
$$ \theta_{i} = (4/9)^{k-i}~.$$
\end{theorem}

In Section~\ref{subsec:running-time-initial-splitting},
we will consider that pseudonorms behave in terms of smoothness like
integers bounded by $N^e$ (instead of $N$). We will need
the following lemmas.

\begin{lemma}[{\cite[\S 4.1]{PKC:ComSem06}, \cite[Lemma 1]{AC:Guillevic15}
    Running time of ${B}$-smooth decomposition of integers with ECM}]
\label{lemma: asympt cpx boot}
Let $N_i$ be integers taken uniformly at random and bounded by $N^e$,
for a fixed real number $e > 0$.
Write ${B} = L_N[\alpha_{B}, \gamma]$ the smoothness bound.
Then the expected running time to obtain a $B$-smooth $N_i$, using ECM
for $B$-smooth tests,
is $L_N[1/3, (3e)^{1/3}]$, obtained with
$B = L_N[2/3, e/c = (e^2/3)^{1/3}]$.
\end{lemma}

\begin{lemma}[{\cite{Pomerance82,PhD:Barbulescu13} Running time
    of ${B}$-smooth decomposition of integers with ECM and $k$-EAS}]
\label{lemma:asympt-cpx-boot-ECM+EAS}
Let $N_i$ be integers taken uniformly at random and bounded by $N^e$,
for a fixed real number $e > 0$.
Write ${B} = L_N[\alpha_{B}, \gamma]$ for the smoothness bound.
Then the expected running time to obtain a $B$-smooth $N_i$, using ECM
 for $B$-smooth tests and Pomerance's Early Abort Strategy with one test,
is $L_N[1/3, c=(3e)^{1/3}(23/27)^{2/3}]$, obtained with
$B = L_N[2/3, e/c]$. 
The expected running time with $k$-EAS is
$L_N[1/3, c= (3e)^{1/3} ((15+4(2/3)^{3k})/19)^{2/3}]~$, with
$B = L_N[2/3, e/c]$. 
\end{lemma}

We will mix Pomerance' strategy with our new initial splitting step to
improve its running time.

\subsection{LLL algorithm}
We recall an important property of the LLL algorithm~\cite{LLL82}
that we will widely use in this paper.
Given a lattice $\mathcal{L}$ of $\ZZ^n$ defined by a basis given by
an $n \times n$ matrix $L$,   
and parameters $\frac{1}{4} < \delta <1$, $\frac{1}{2} < \eta <
\sqrt{\delta}$, the LLL algorithm outputs a $(\eta,
\delta)$-\emph{reduced basis} of the lattice. 
the coefficients of the first (shortest) vector are bounded by
$$(\delta - \eta^2)^{\frac{n-1}{4}} \det(L)^{1/n} ~.$$
In the remainder of this paper, we will simply denote by $C$ this LLL
approximation factor.

\subsection{NFS and Tower variants}

\subsubsection{Settings}
There exist many polynomial selection methods to initialize the
NFS algorithm for large and medium characteristic finite fields. 
We give in 
Table~\ref{tab: deg size polyselect comparison} the properties
of the polynomials that we need (degree and coefficient size) to
deduce an upper bound of the pseudonorm, as in \eqref{eq:norm bound},
and \eqref{eq:ExTNFS-pseudonorm}.

\begin{figure}[h]
  \centering
  \subfloat[{NFS number fields}]
  {
    \label{subfig:NFS-tower}
    \begin{tikzpicture}[>=latex]
      \node (QQ) {$\QQ$};
      \node[above of=QQ, node distance=3em] (Kh) {};
      \node[above left of=Kh, node distance = 5em] (Kf) {$K_{f_0}$};
      \node[above right of=Kh, node distance = 5em] (Kg) {$K_{f_1}$};
      \draw[-] (QQ) -- node[left]{$\deg f_0 \geq n$} (Kf);
      \draw[-] (QQ) -- node[right]{$\deg f_1 \geq n$} (Kg);
    \end{tikzpicture}
  }
  \subfloat[{Tower-NFS number fields}]
  {
    \label{subfig:ExTNFS-tower}
    \begin{tikzpicture}[>=latex]
      \node (QQ) {$\QQ$};
      \node[above of=QQ, node distance=3em] (Kh) {$K_h$};
      \draw[-] (QQ) -- node[right]{$\deg h = n_1$} (Kh);
      \node[above left of=Kh, node distance = 5em] (Kf) {$K_{f_0}$};
      \node[above right of=Kh, node distance = 5em] (Kg) {$K_{f_1}$};
      \draw[-] (Kh) -- node[left]{$\deg f_0 \geq n_2$} (Kf);
      \draw[-] (Kh) -- node[right]{$\deg f_1 \geq n_2$} (Kg);
    \end{tikzpicture}
  }
  \caption{Extensions of number fields for NFS and tower variants}
  \label{fig:NFS-ExTNFS-towers}
\end{figure}
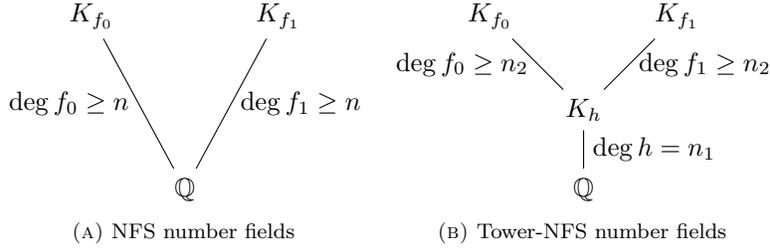

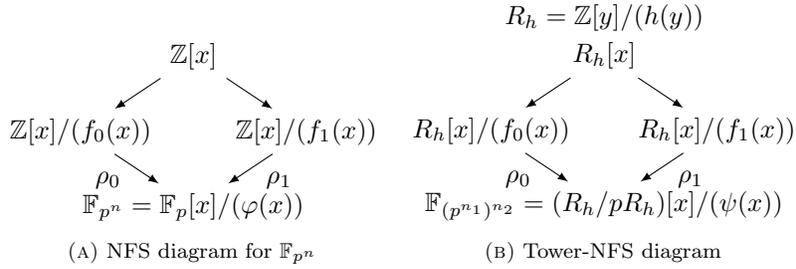
\begin{figure}[h]
  \centering
  \subfloat[{NFS diagram for $\FF_{p^n}$}]
  {
    \label{subfig:NFS-diag}
    \begin{tikzpicture}[>=latex]
    \node (Zx) at (1.5, 2) {$\ZZ[x]$};
    \node (Kf) at (0, 1) {$\ZZ[x]/(f_0(x))$};
    \node (Kg) at (3, 1) {$\ZZ[x]/(f_1(x))$};
    \node (Fpn) at(1.5, 0) {$\FF_{p^n} = \FF_p[x]/(\varphi(x))$};
    \draw[->] (Zx) -- (Kf);
    \draw[->] (Zx) -- (Kg);
    \draw[->] (Kf) -- node[yshift=-1ex,xshift=-0.25em,left]  {$\rho_0$} (Fpn);
    \draw[->] (Kg) -- node[yshift=-1ex,xshift=0.25em,right]  {$\rho_1$}(Fpn);
  \end{tikzpicture}
  }
  \subfloat[{Tower-NFS diagram}]
  {
    \label{subfig:ExTNFS-diag}
    \begin{tikzpicture}[>=latex]
      \node (defR) at (1.5, 2.5) {$R_h=\ZZ[y]/(h(y))$};     
    \node (Rx) at (1.5, 2) {$R_h[x]$};
    \node (Kf) at (0, 1) {$R_h[x]/(f_0(x))$};
    \node (Kg) at (3, 1) {$R_h[x]/(f_1(x))$};
    \node (Fpn) at(1.5, 0) {$\FF_{(p^{n_1})^{n_2}} = (R_h/pR_h)[x]/(\psi(x))$};
    \draw[->] (Rx) -- (Kf);
    \draw[->] (Rx) -- (Kg);
    \draw[->] (Kf) -- node[yshift=-1ex,xshift=-0.25em,left]  {$\rho_0$} (Fpn);
    \draw[->] (Kg) -- node[yshift=-1ex,xshift=0.25em,right]  {$\rho_1$}(Fpn);
  \end{tikzpicture}    
  }  
    \caption{\label{fig:diag}NFS and tower variant diagrams for $\FF_{p^n}$}
\end{figure}

Three polynomials define the NFS setting: $\psi, f_0, f_1$, where
$f_0, f_1$ are two polynomials of integer coefficients, irreducible
over $\QQ$, of degree $\geq n$, defining two non-isomorphic number
fields, and whose GCD modulo $p$ is an irreducible polynomial
$\psi$ of degree $n$, used to define the extension $\F_{p^n} =
\F_p[x]/(\psi(x))$. 

In a tower-NFS setting, one has $n = n_1n_2$, $n_1,n_2 \neq 1$ and four
polynomials are defined: $h,\psi,f_0,f_1$, where $\deg
h = n_1$ and $h$ is irreducible modulo $p$, $\deg \psi = n_2$ and $\psi$ is irreducible
modulo $p$, and $\gcd(f_0 \bmod (p, h), f_1 \bmod (p, h)) = \psi$.
It can be seen as a generalization of the NFS setting as follows: writing
$n=n_1n_2$, one starts by defining a field extension $\F_{p^{n_1}} =
\F_p[y]/(h(y))$ and then adapting any previously available polynomial
selection designed for NFS in $\GF(p^{n_2})$, using $\F_{p^{n_1}}$ as the base field
instead of $\F_{p}$. When $\gcd(n_1,n_2)>1$, the polynomials
$f_0,f_1$, resp., $\psi$, will have coefficients in $\QQ[y]/(h(y))$,
resp., $\F_{p^{n_1}}$, instead of $\QQ$, resp., $\F_p$.
Then one defines the second
extension $\F_{p^n_1}[x]/(\psi(x))$ of degree $n_2=\deg_x\psi$.

Again, to cover all the cases, we consider $\F_{p^n} =
\F_{(p^{n_1})^{n_2}}$. The NFS case will correspond to $n_1=1$,
$n_2=n$ and the original TNFS case to $n_1=n$, $n_2=1$.
% put here the diagrams

\subsubsection{Pseudonorm and upper bound}
\label{subsec: norm bound}
Let $f$ be a monic irreducible polynomial over $\QQ$, and let $K =
\QQ[x]/(f(x))$ be a number field. Write $T \in K$ as a polynomial in $x$: $T
= \sum_{i=0}^{\deg f-1} a_i x^i$. 
The norm is defined by a resultant computation:
\begin{equation}
\Norm_{K/\QQ}(T) =  \Reslt(f, T) ~.
\end{equation}
In the NFS case, we will consider elements expressed as polynomials in $x$
whose coefficients are integers. We define the pseudonorm as the resultant of the element with
the given polynomial $f$:
$$ T= \sum_{i=0}^{\deg f-1}a_ix^i,~\mbox{pseudonorm}(T(x)) = \Reslt(T(x),f(x)) ~.$$
We use Kalkbrener's bound \cite[Corollary 2]{Kalkbrener97} for an upper bound:
\begin{equation}\label{eq:Kalkbrener-bound}
|\Reslt(f,T)|\leq \kappa(\deg f,\deg T) \norm{f}_\infty^{\deg T}
\norm{T}_\infty^{\deg f},
\end{equation}
where $\kappa(n,m)=\binom{n+m}{n}\binom{n+m-1}{n}$ and
$\norm{f}_\infty = \max_{0 \leq j \leq \deg f} |f_j|$ is the absolute
value of the largest coefficient.
An upper bound for $\kappa(n,m)$ is $(n+m)!$.
We will use the following bound in Section~\ref{sec:app-large-char}:
\begin{equation}\label{eq:norm bound}
\Norm_{K_f/\QQ}(T) \leq (\deg f + \deg T)! \norm{f}_\infty^{\deg T}
\norm{T}_\infty^{\deg f} ~.
\end{equation}
In a Tower-NFS case, we nest two resultants:
$$ T= \sum_{i=0}^{\deg f-1}\sum_{j=0}^{\deg h-1}a_{ij}y^jx^i,~\mbox{pseudonorm}(T(x,y)) = \Reslt_y(\Reslt_x(T(x),f(x)),h(y)) ~.$$
A bound is \cite[\S A Lemma 2]{C:KimBar16} 
\begin{equation}
  \label{eq:ExTNFS-pseudonorm}
  \begin{array}{l}
| N_{K_f/\QQ} \sum_{i=0}^{\deg_x\! P}\sum_{j=0}^{\deg
    h-1}a_{ij}\alpha_h^j\alpha_{f}^i | \\
    <
\|a_{ij}\|_{\infty}^{\deg h \deg f} \|f\|_{\infty}^{\deg_x\!P\deg h}
\|h\|_\infty^{(\deg_x\!P+\deg f)(\deg h-1)} D(\deg h,\deg f)~,
  \end{array}
  \end{equation}
where $\|a_{ij}\|_{\infty} = \max_{i,j} | a_{ij}|$ and $D(d_1,d_2)$ is a combinatorial term, $D(d_1,d_2) =
((2d_2-1)(d_1-1)+1)^{d_1/2}(d_1+1)^{(2d_2-1)(d_1-1)/2} ((2d_2-1)!d_1^{2d_2})^{d_1}$.

\section{Faster Initial Splitting with NFS and Tower variants for medium and large characteristic finite fields}
\label{sec:app-large-char}

We apply Algorithm~\ref{alg:lower-degree-poly} to the medium and large
characteristic cases.
For a general exposition, we assume that we are in a tower setting,
where $Q = p^n = (p^{n_1})^{n_2}$.
The elements of $\F_{p^n}$ are represented as $T = 
\sum_{i=0}^{n_1-1} \sum_{j=0}^{n_2-1} a_{i,j}y^j x^i$.
the NFS setting corresponds to $n_1=1$, $n_2=n$. When $n$ is prime,
the tower setting is $n_1=n$, $n_2=1$ but our algorithm does not apply.
Denote by $h(y)$ the polynomial defining the field $\F_{p^{n_1}}$ and by $\psi$
the polynomial defining the degree $n_2$ extension $\F_{(p^{n_1})^{n_2}}$.
Here we are not interested (only) in computing a preimage of degree 
as small as possible, but more generally one whose size of pseudonorm
is as small as possible.
According to the bounds \eqref{eq:norm bound}, \eqref{eq:ExTNFS-pseudonorm},
we need to combine small coefficients $a_{i,j}$ (to reduce the contribution
of $\norm{a_{ij}}_\infty^{\deg h \deg f}$) with a small degree in $x$ (to reduce the
contribution of $\norm{f}_\infty^{\deg_x P\deg h}$), and
balance the two terms to find a pseudonorm of smaller size.

\subsection{The algorithm}
We start again with the same idea as in Algorithm~\ref{alg:lower-degree-poly}:
let $d$ be the largest proper divisor of $n$, with $1 < d< n$.%
\footnote{$d=\deg(h)=n_1$ is
the case studied independently in the preprint \cite{EPRINT:ZZLL16}.
Since an earlier version of this work was presented at Asiacrypt 2015 and ECC
2015, and the question of how to use larger subfields raised in discussions
at these conferences, it is not surprising that other researchers though of similar ideas to
improve individual discrete logarithms in the same time period.}
 Assume we want to
obtain a preimage $P \in \F_{p^{n_1}}[x]$ of the target, of degree $(n-d)/n_1 \leq \deg P < \deg f$.
We will use relations of the form
\begin{equation*}
    P = u T \pmod{\psi}, \mbox{ where } u^{p^d-1}=1~~ \mbox{ as in } \eqref{eq:P-eq-uT}.
\end{equation*}
We use the relations
$$ x^iy^j p = 0 \pmod{p,h,\psi} \mbox{ for } 0 \leq ij < d~,$$
$$ P = U^i T  \pmod{p,h,\psi} ~,$$
where $\{1,U,\ldots,U^{d-1}\}$ is a polynomial basis of $\F_{p^d}$ and where
$$ x^iy^j \psi = 0 \pmod{p,h,\psi} \mbox{ for } 0\leq j < n_1, ~0 \leq i < \deg(P)-n_2.$$
We define the lattice of these relations and we obtain a matrix

$$L_{n_1(\deg P+1) \times n_1(\deg P+1)} =\begin{bmatrix}
p & \\
  & \ddots & \\
  &        & p & \\
  \multicolumn{4}{r}{\cff(T)} \\
  \multicolumn{4}{r}{\cff(UT)} \\
  &        & \vdots \\
  \multicolumn{4}{r}{\cff(U^{d-1} T)} \\
  & & \multicolumn{3}{r}{\cff(\psi)} \\
  & & \multicolumn{3}{r}{\ddots} \\
  \multicolumn{6}{r}{\cff(x^i(y^j \psi \bmod h(y)))}
\end{bmatrix}$$
We want to obtain a matrix in row-echelon form. The $d$ first rows and 
the $n_1(\deg P - n_2)$ last rows are in row-echelon form by construction.
We compute Gaussian elimination to obtain a reduced row-echelon form for the
rows $U^iT$. We use $\F_p$-linear 
combinations of these rows, and we allow divisions in $\F_p$ so that the leading
coefficient is one. We then obtain a square matrix of dimension $n_1
(\deg P+1)$ in row-echelon form. Now at this point we apply a lattice reduction
algorithm such as LLL or BKZ to reduce the size of the coefficients of $L$.
We obtain a matrix $R$ whose first row has coefficients bounded by 
$C_{\LLL}\det (L) ^{1/(n_1(\deg P+1))} = p^{(n-d)/(n_1(\deg P+1))}$.

\begin{algorithm}[htbp]
  \DontPrintSemicolon
  \caption{Initial splitting,
    Tower-NFS setting}
  \label{alg:boot-d-subfield-ExTNFS}
    \KwIn{Finite field $\F_{p^n}$, $n=n_1n_2$,
    monic irreducible polynomials $h,\psi$ s.t. $\F_{p^{n_1}} = \F_p[y]/(h(y))$, $\F_{(p^{n_1})^{n_2}} = \F_{p^{n_1}}[x]/(\psi(x))$, 
    prime order subgroup $\ell \mid \Phi_n(p)$, 
    generator $g$ (of the order $\ell$ subgroup), 
    target $T_0 \in \F_{p^n}$, degree of the preimage $\deg P$, polynomial $f_i$,
    smoothness bound $B_1$}
  \KwOut{$t \in \{1, \ldots, \ell-1\}$, $P \in \ZZ[x]$
    s.t.~$\log_g \rho(P) \equiv t + \log_g T_0$, and the pseudonorm $\Reslt_y(\Reslt_x(P,f_i),h)$ is
    $B_1$-smooth}
  $d \gets$ the largest divisor of $n$, $1 \leq d < n$ \;
  Compute a polynomial basis  $(1, U, U^2, \ldots, U^{d-1})$ of the
  subfield $\F_{p^d}$, where $U$ satisfies $U^{p^d-1}=1 \in \F_{p^n}$\;
  \Repeat{$\Reslt_y(\Reslt_x(P,f_i),h)$ is $B_1$-smooth}{
    take $t \in \{1, \ldots, \ell-1\}$ uniformly at random \;
    $T \gets g^t T_0 \in  \F_{p^n}$ \;
    $L \gets
    \begin{bmatrix}
p & \\
  & \ddots & \\
  &        & p & \\
  \multicolumn{4}{r}{\cff(T)} \\
  \multicolumn{4}{r}{\cff(UT)} \\
  &        & \vdots \\
  \multicolumn{4}{r}{\cff(U^{d-1} T)} \\
  & & \multicolumn{3}{r}{\cff(\psi)} \\
  & & \multicolumn{3}{r}{\ddots} \\
  \multicolumn{6}{r}{\cff(x^i(y^j \psi \bmod h(y)))}
\end{bmatrix}$ \;
Compute a $\F_p$- reduced row echelon form of the rows $n-d+1$ to $n$ of $L$\;
$N \gets \mathtt{LatticeReduction}(L)$ \; 
$P \gets$ polynomial in $\ZZ[y,x]$ made of the shortest vector output by the \texttt{LatticeReduction} algorithm \;
}(\tcp*[f]{ECM, ECM+EAS, or ECM+$k$-EAS})
\Return {$t$, $P$, factorization of $\Reslt_y(\Reslt_x(P, f_i),h)$}
  \end{algorithm} 

\subsection{Properties and pseudonorm size bound}

\begin{proposition}\label{prop:log-equality-mod-ell}
The preimage $P$ output by Algorithm~\ref{alg:boot-d-subfield-ExTNFS}
satisfies $\log_g \rho(P) \equiv \log_g g^t T_0 = \log_g T_0 + t
\bmod \Phi_n(p)$, where $\rho: \ZZ[x,y] \to \F_{(p^{n_1})^{n_2}}$ was defined in Figure~\ref{fig:diag}. 
\end{proposition}

\begin{proof}[Proof of Proposition~\ref{prop:log-equality-mod-ell}]
Each row of the row-echelon matrix $M$ represents a $\F_p$-linear
combination of the $d$ elements $U^iT$, $0 \leq i \leq d-1$,
i.e., an element $\sum_{i=0}^{d-1}  \lambda_i U^i
T$, where $\lambda_i \in \Fp$. We can factor $T$ in the
expression. Each element $u_j=\sum_{i=0}^{d-1}  \lambda_i U^i$ satisfies
$u_j^{p^d-1}=1$, i.e., is in $\F_{p^d}$ by
construction. So each row represents an element $T_j = u_j T$, where
$u_j^{p^d-1}=1$ ($u_j \in \F_{p^d}$), so that $\log T_j \equiv \log T
\bmod \Phi_n(p)$ by
Lemma~\ref{lemma:log equality up to subgroup elt}. 

The second part of the proof uses the same argument: the short vector
output by the LLL algorithm is a linear combination of the rows of the
matrix $N$. Each row represents either 0 or a $\F_{p^d}$-multiple  $T_j$ of
$T$, hence the short vector is also a $\F_{p^d}$-multiple of $T$.
We conclude thanks to Lemma~\ref{lemma:log equality up to subgroup elt}, 
that $\log \rho(P) \equiv \log T \bmod \Phi_n(p)$.
\end{proof}

\begin{proposition}\label{prop:pseudonorm-bound-ExTNFS}
  The pseudonorm of $P$ in Algorithm~\ref{alg:boot-d-subfield-ExTNFS}
  has size
  \begin{equation}
    \label{eq:pseudonorm-bound-ExTNFS}
|\Reslt_y(\Reslt_x(P, f_i),h)| = O\left( Q^{(1-\frac{d}{n})\frac{\deg f_i}{\deg_x
      P+1}} \|f_i\|_\infty ^{n_1 \deg_x P} \right)
  \end{equation}
assuming that $\|h\|_\infty = O(1)$. 
\end{proposition}
\begin{proof}[Proof of Proposition~\ref{prop:pseudonorm-bound-ExTNFS}.]
The matrix $N$ computed in Algorithm~\ref{alg:boot-d-subfield-ExTNFS}
is a square matrix of $(\deg_xP+1)n_1$ rows
and columns, whose coefficients are in $\F_p$. Its determinant is
$\det N = p^{n-d} = Q^{1-d/n}$. Using the LLL algorithm for the
lattice reduction, the coefficients of the
shortest vector $P$ are bounded by $C Q^{(1-d/n)/((\deg_x P+1)n_1)}$,
where $C$ is the LLL factor.
We obtain the bound~\eqref{eq:pseudonorm-bound-ExTNFS}
according to the bound formula
\eqref{eq:ExTNFS-pseudonorm}, and neglecting the combinatorial factor
$D(n_1, \deg f_i)$. Moreover in the Tower-NFS setting, the polynomial selection is designed such
that $\|h\|_\infty = O(1)$.
\end{proof}

We finally obtain the following.
\begin{theorem} \label{th: smaller norm bound}
Let $\GF(p^n)$ be a finite field, and let $d$ be the largest divisor of $n$, $d<n$, and $d=1$ if $n$ is prime.
Let $n=n_1n_2$ and $h,\psi,f_i$ be given by a polynomial selection method.
Let $T \in \FF_{(p^{n_1})^{n_2}}$ be an element which is not in a proper subfield
of $\FF_{p^n}$.
Then there exists a preimage $P \in \ZZ[x,y]$ of $T$,
of any degree (in $x$) between $\lfloor n_2-d/n_1 \rfloor$ and $\deg f_i-1$, of coefficients bounded
by $O(Q^{(1-\frac{d}{n})  \frac{1}{(\deg P +1)n_1}})$,
and such that when $P$ is mapped in $\F_{(p^{n_1})^{n_2}}$ as $\rho(P)$, its discrete
logarithm is equal to the discrete logarithm of $T$ modulo $\Phi_n(p)$ (and in particular
modulo any prime divisor $\ell$ of $\Phi_n(p)$), that is,
$$\log \rho(P) \equiv \log T \bmod \Phi_n(p)~.$$
The degree of $P$ in $x$ and the polynomial $f_i$ can be chosen to
minimize the resultant (pseudonorm):
$$ \min_{i} \min_{\lfloor n_2-d/n_1\rfloor \leq \deg_x P \leq
  \deg f_i-1 } \|f_i\|_{\infty}^{n_1\deg_x P} Q^{(1-\frac{d}{n})
  \frac{\deg f_i}{\deg_x P +1}} ~.$$
\end{theorem}

We recall in Table~\ref{tab: deg size polyselect comparison}
the degree and coefficient sizes of the
polynomial selections published as of July 2017.

\begin{corollary}\label{corl:sizes-pseudonorms}
  With the notation of Table~\ref{tab: deg size polyselect comparison}
  and the NFS setting corresponding to $n_2=n$ and $n_1=1$,
  \begin{enumerate}
\item For the polynomial selection methods where there is a side $i$ such that
$\|f_i\|_\infty = O(1)$ (GJL, Conjugation, Joux--Pierrot and
Sarkar--Singh up to now), we do the initial splitting on this side and choose
$\deg_x P = \deg_x f_i-1$ to obtain the smallest norm: 
$ |\Reslt_y(\Reslt_x(P, f_i),h)| = O\left( Q^{1-\frac{d}{n}}
\right)$. We obtain the same bound for NFS and its tower variants.
  \item When $\| f_i\|_\infty = Q^{1/(2n)}$ as for the \JLSVi{}
    method, the bound is \linebreak
    $Q^{(1-\frac{d}{n})\frac{n_2}{\deg_x P+1}+\frac{\deg_x P}{2 n_2}}$.
    When $\deg_x P = \deg_x f_i - 1 = n_2-1$, one obtains
    $Q^{\frac{3}{2} - \frac{d}{n} - \frac{1}{2n_2}}$.
    In the NFS setting, $n_2=n$, while in the tower setting, $n_2< n$
    and the pseudonorm is slightly smaller.
  \item When $\| f_i\|_\infty = Q^{1/(n_1(D+1))}$ as  for the \JLSVii{} method,
    the lower bound is
    $Q^{\frac{\deg_x P}{D+1} + (1-\frac{d}{n})\frac{n_2}{\deg_x P+1}}$ on the $f_0$-side
    where $\deg f_0 = n_2$, and it is \linebreak
    $Q^{\frac{\deg_x P}{D+1} + (1-\frac{d}{n})\frac{D}{\deg_x P+1}}$ on
    the $f_1$-side, where $\deg f_1 = D \geq n_2$.
    According to the value of $n$, one can decide which value of
    $\deg_x P$ will produce a smaller norm.
\end{enumerate}
\end{corollary}

\begin{table}[hbtp]
\caption{Properties: degree and coefficient size of the main
  polynomial selection methods for NFS-DL in $\FF_{Q}$, where $Q=p^n$.}
We give a bound on the coefficient size of the polynomials with the notation $\|f_i\|_\infty = O(x)$.
 To lighten the notation, we only write $x$, without $O()$.
 In the Joux--Pierrot method, the prime
  $p$ can be written $p=p_x(x_0)$, where $p_x$ is a polynomial of tiny
  coefficients and degree at least $2$. This table takes into account the
  methods published until July 2017.
\centering
  
\label{tab: deg size polyselect comparison}
\begin{tabular}{|@{\hspace*{2pt}}c@{\hspace*{2pt}}||@{\hspace*{2pt}}c@{\hspace*{2pt}}|c | c |@{\hspace*{2pt}}c@{\hspace*{2pt}}| c|}
\hline method & $\deg h$ & $\deg f_0$ & $\deg f_1$ & $\|f_0\|_\infty$ & $\|f_1\|_\infty$ \\
\hline    NFS & \mc{5}{c|}{} \\
\hline
JLSV\textsubscript{1} \cite{C:JLSV06} &  & $n$        &  $n$    & $Q^{1/2n}$  & $Q^{1/2n}$    \\
JLSV\textsubscript{2} \cite{C:JLSV06} &  & $n$        &  $D > n$& $Q^{1/(D+1)}$& $Q^{1/(D+1)}$ \\
GJL \cite{Mat06,PhD:Barbulescu13,EC:BGGM15}&&$D+1$&$D\geq n$& $\log p$   & $Q^{1/(D+1)}$  \\
Conjugation \cite{EC:BGGM15}      &  & $2n$       & $n$     & $\log p$   & $Q^{1/2n}$     \\
  \begin{tabular}{@{}c@{}}
    Joux-Pierrot \cite{PAIRING:JouPie13}\\
    $p=p_x(x_0)$ 
  \end{tabular}                   &  & $n(\deg p_x)$& $n$   & $\log p$   & $Q^{1/(n \deg p_x)}$ \\
  \begin{tabular}{@{}c@{}}
    Sarkar-Singh \cite{EC:SarSin16}\\
    $n=n_1n_2$, $D \geq n_2$
  \end{tabular}
                                  &  & $(D+1)n_1$ & $D n_1$ & $\log p$   & $Q^{1/(n_1(D+1))}$ \\
\hline Tower-NFS & \mc{5}{c|}{} \\
\hline TNFS + base-$m$ \cite{AC:BarGauKle15} 
                                  &$n$& $D$       & $1$     & $p^{1/D}$  & $p^{1/D}$ \\
\hline
  \begin{tabular}{@{}c@{}}
    Tower-JLSV\textsubscript{1}\\
    $n=n_1n_2$
  \end{tabular}
            & $n_1$ & $n_2$ & $n_2$ & $Q^{1/(2n)}$ & $Q^{1/(2n)}$ \\
  \begin{tabular}{@{}c@{}}
    Tower-JLSV\textsubscript{2} \\
    $n=n_1n_2$ \cite{EPRINT:Kim15,C:KimBar16}
  \end{tabular}
            & $n_1$ & $n_2$ & $D \geq n_2$ & $Q^{1/(n_1(D+1))}$ & $Q^{1/(n_1(D+1))}$ \\
  \begin{tabular}{@{}c@{}}
    Tower-GJL \\
    $n=n_1n_2$ \cite{C:KimBar16}
  \end{tabular}
            & $n_1$ & $D+1$ & $D\geq n_2$ & $\log p$ & $Q^{1/(n_1(D+1))}$ \\
  \begin{tabular}{@{}c@{}}
    Tower-Conjugation \\
    $n=n_1n_2$ \cite{EPRINT:Barbulescu15,C:KimBar16,PKC:KimJeo17} \\
%n_1 = O(\frac{1}{12^{\frac{1}{3}}}(\frac{\log Q}{\log \log Q})^{1/3})
  \end{tabular}
            & $n_1$ & $2n_2$ & $n_2$ & $\log p$ & $Q^{1/(2n)}$ \\
  \begin{tabular}{@{}c@{}}
    Tower-Joux--Pierrot \\
    $n=n_1n_2$, $p=p_x(x_0)$ \\
    \cite{C:KimBar16,PKC:KimJeo17} \\
\end{tabular}
            & $n_1$ & $n_2(\deg p_x)$ & $n_2$ & $\log p$ & $Q^{1/(n \deg p_x)}$ \\
  \begin{tabular}{@{}c@{}}
    Tower-Sarkar--Singh\\
    $n=n_1n_2n_3$, $D\geq n_3$ \\
    \cite{AC:SarSin16,EPRINT:SarSin16:401,EPRINT:SarSin16:537}
  \end{tabular}
            & $n_1$ & $(D+1)n_2$ & $Dn_2$ & $\log p$ & $Q^{1/(n_1 n_2(D+1))}$ \\
  \hline
\end{tabular}
\end{table}

%\FloatBarrier
\subsection{running time}
\label{subsec:running-time-initial-splitting}

To apply Lemma~\ref{lemma: asympt cpx boot} to the initial splitting case, we make the
usual heuristic assumption that the pseudonorms of the elements
$g^tT_0$ behave asymptotically like random integers of the same
size. Their size is $O(Q^e)$, so we replace $N^e$ by $Q^e$.
The basis $\{1,U,\ldots, U^{d-1}\}$ can be precomputed.
The cost of computing the $U^iT$ for $0\leq i \leq d-1$ is at most
$dn^2$ multiplications in $\Fp$ with a schoolbook multiplication
algorithm. We can roughly upper-bound it by $O(n^3)$.
The time needed to compute the reduced row-echelon form of a $d \times n$ matrix is in
$O(n^3)$ which is polynomial in $n$ \cite{HFF:DumPer13}. These two complexities are
asymptotically negligible compared to any $L_Q[\alpha > 0]$.
We obtain the following.
\begin{corollary}\label{corl:running-time}
  The running time of the initial splitting step with
  Algorithm~\ref{alg:boot-d-subfield-ExTNFS}
  to find a $B$-smooth pseudonorm, where the
  pseudonorm has size $O(Q^e)$ for a fixed real number $e>0$
  determined by the polynomial selection
  (Table~\ref{tab: deg size polyselect comparison}, two right-most columns), is
  \begin{enumerate}
  \item $L_Q[1/3, c=(3e)^{1/3}]$ with ECM to perform the smoothness
    tests;
  \item $L_Q[1/3, c=(3e)^{1/3}(23/27)^{2/3}]$ with ECM and EAS;
  \item $L_Q[1/3, c=(3e)^{1/3} ((15+4(2/3)^{3k})/19)^{2/3}]$ with ECM
    and $k$-EAS.
  \end{enumerate}
  For each case, the lower bound was obtained for $B=L_Q[2/3, e/c]$.
\end{corollary}

Corollary~\ref{corl:sizes-pseudonorms} gives a bound on the size of
the pseudonorms, from which we can deduce $e$ to apply
Corollary~\ref{corl:running-time}, and get the expected running time.

\section{Examples}
\label{sec:large-char-examples}
\begin{example}\label{ex:Fp6}
  Let $p = \lfloor 10^{25}\pi \rfloor + 7926 =
  \coeffs{31415926535897932384634359}$ be a 85-bit prime made of the
  first 26 decimals of $\pi$
  so that $\F_{p^6}$ is a 509-bit finite field. Moreover, $\Phi_6(p)
  = p^2 - p + 1$ is a 170-bit prime, we denote it by 
  $\ell = \coeffs{986960440108935861883947021513080740536833738706523}$. 
We want to compute discrete logarithms in
  the order-$\ell$ cyclotomic subgroup of $\F_{p^6}$. 
The \JLSVi{} method computes two polynomials $f_0,f_1$, where $\deg f_0 = \deg f_1
= 6$, and $\|f_i\|_\infty \approx p^{1/2}$. In our example, we
have $\log_2 \|f_0\|_\infty = 44.67$ and $ \log_2 \|f_1\|_\infty = 46.67$
(and $\log_2 p/2 = 42.35$):
$$ \begin{array}{@{}r@{\hspace*{3pt}}l@{}}
f_0 =& x^6 - \coeffs{11209975711932\ } x^5 - \coeffs{28024939279845\ } x^4 - \coeffs{20\ } x^3 \\
     &  + \coeffs{28024939279830\ } x^2 + \coeffs{11209975711938\ } x + 1 \\
f_1 =& \coeffs{5604994576830\ } x^6 + \coeffs{20986447533158\ } x^5 - \coeffs{31608799819555\ } x^4 \\
     & - \coeffs{112099891536600\ } x^3  - \coeffs{52466118832895\ } x^2 + \coeffs{12643519927822\ } x \\
     & + \coeffs{5604994576830}. \\  
\end{array}$$
Since $f_0$ is already of degree $6$ and monic, it can define the
extension $\F_{p^6} = \F_p[x]/(f_0(x))$.
Let $T_0$ be our target in $\F_{p^6}$ whose coefficients are made
of the decimals of $\pi$ (starting at the 26-th decimal, since the
first 25 ones were already used for $p$):
$$\begin{array}{@{}r@{\hspace{3pt}}c@{\hspace{3pt}}l@{\hspace{3pt}}}
  T_0 &=& \coeffs{6427704988581508162162455\ } x^5 + \coeffs{16240052432693899613177738\ } x^4 \\
& &  + \coeffs{4509390283780949909020139\ } x^3 + \coeffs{3868374359445757647591444\ } x^2 \\
& &  + \coeffs{8209755913602112920808122\ } x + \coeffs{3279502884197169399375105}. \\
\end{array}$$
Let $g = x+3$ be a generator of $\F_{p^6}$.
Let $(1, U, U^2)$ be a polynomial basis of $\F_{p^3}$ considered as an implicit
subfield of $\F_{p^6}$, where $U = g^{1+p^3} =
\Norm_{\F_{p^6}/\F_{p^3}}(g)$. 
We run Algorithm~\ref{alg:boot-d-subfield-ExTNFS} and find that the
fourth preimage of $T=g^{812630}T_0$ gives a 61-smooth
pseudonorm. 
We compute the reduced row-echelon form
$$M = \begin{bmatrix}
  m_{00} & m_{01} & m_{02} & 1 & 0 & 0 \\
  m_{10} & m_{11} & m_{12} & m_{13} & 1 & 0 \\
  m_{20} & m_{21} & m_{22} & m_{23} & m_{24} & 1 \\
\end{bmatrix}
\mbox{ of the matrix }
\begin{bmatrix}
  T \\
  UT \\
  U^2 T \\
\end{bmatrix}~,$$
where
$$\begin{array}{ll}
m_{00}= \coeffs{30930778358987253373198053} & m_{01}= \coeffs{16172276732961477886471865},\\
m_{02}= \coeffs{251875570676859576731124}   & m_{10}= \coeffs{8981071706647180870633008}, \\
m_{11}= \coeffs{26297121233008662476505921} & m_{12}= \coeffs{4999545867425989707589927}, \\
m_{13}= \coeffs{4380553940470247124926451}  & m_{20}= \coeffs{4787502941827866787698085}, \\
m_{21}= \coeffs{18855419729462744536987506} & m_{22}= \coeffs{15450347628775338768673252}, \\
m_{23}= \coeffs{31092163492444411597011243} & m_{24}= \coeffs{9824382756181109886988461}. \\
\end{array}$$
 Then we reduce with the LLL algorithm the following lattice defined by the $(6\times
6)$-matrix, where $m_{ij}$ stands for the coefficient at row $i$ and column $j$
of the above matrix $M$, and $m_{i,3+i}=1$:
$$N = \begin{bmatrix}
  p & 0 & 0 & 0 & 0 & 0 \\
  0 & p & 0 & 0 & 0 & 0 \\
  0 & 0 & p & 0 & 0 & 0 \\
  m_{00}  & m_{01} & m_{02} & 1 & 0 & 0 \\
  m_{10}  & m_{11} & m_{12}  &  m_{13} & 1 & 0 \\
  m_{20}  & m_{21} & m_{22}  & m_{23} &  m_{24} & 1 \\
\end{bmatrix}~.$$
Each row of LLL$(N)$ gives us a
preimage $P\in \ZZ[x]$ of short coefficients
such that $\log_2 \|P\|_\infty \approx \log_2 p/2 = 42.34$ bits
and $\log \rho(P) \equiv \log T \bmod \ell$ (in other words,
$(T/\rho(P))^{\frac{p^6-1}{\ell}} = 1$).
The fourth row has coefficients of at most $41.82$ bits and gives
$$\begin{array}{@{}r@{\hspace{3pt}}l@{}}
  P =& \coeffs{482165402365\ } x^5 + \coeffs{3892831179802\ } x^4 + \coeffs{2694050932529\ } x^3 \\
     & + \coeffs{2325450478817\ } x^2 + \coeffs{1117470283668\ } x + \coeffs{3688595236671\ }.
\end{array}$$
The pseudonorm of $P$ w.r.t.~$f_0$ is
$$ \begin{array}{@{}l@{}}
\Reslt(P, f_0) = \\
    \coeffs{32601551184187978602887820222780280368556791213406352787959859478882009
\backslash} \\
  \coeffs{89411710052105812763285379877699363515358275429392312189582741360186561}
\end{array} $$
of 471 bits, which is very close to $\log_2 Q^{11/12} = 466$ bits.
Its factorization in prime ideals of $K_{f_0}$ is
$$\begin{array}{l}
\langle \coeffs{3},x + \coeffs{2}\rangle^3
\langle \coeffs{11},x + \coeffs{5} \rangle
\langle \coeffs{17},x + \coeffs{4} \rangle
\langle \coeffs{67},x + \coeffs{44} \rangle
\langle \coeffs{2011},x + \coeffs{463} \rangle
\langle \coeffs{501997},x + \coeffs{18312} \rangle\\
\langle \coeffs{340575947},x + \coeffs{27999767} \rangle
\langle \coeffs{506032577},x + \coeffs{177467846} \rangle
\langle \coeffs{604579099},x + \coeffs{309800481} \rangle \\
\langle \coeffs{1402910243559283}, x + \coeffs{1034551157262971} \rangle\\
\langle \coeffs{1587503571970639},x + \coeffs{524543605465730} \rangle\\
\langle \coeffs{36834399852305717},x + \coeffs{24916507207930752} \rangle\\
\langle \coeffs{242270403627311729},x + \coeffs{170018299727614229} \rangle\\
\langle \coeffs{1070632553963863603},x + \coeffs{408232161861505290} \rangle\\
\langle \coeffs{4305864084909925127},x + \coeffs{3252872861595329896} \rangle.
\end{array}$$
A common choice for the factor basis would be to set its smoothness
bound to 30 or 32 bits. There are six prime ideals whose norm is larger than 30 bits, and that
should be retreated to reach the factor basis.
This initial splitting, testing all pseudonorms obtained for
$g^iT_0$, $i$ from 0 to 930000, that is, $5.58\cdot 10^6$ pseudonorms,
with our Magma implementation, took
0.95 day on one node of 16 physical cores (32 virtual cores thanks to hyperthreading)
Intel Xeon E5-2650 at 2.0GHz, that is, 15.2 core-days.
\end{example}

\begin{example}[A more general example with NFS]
Assume that $n$ is even and let $T \in \F_{p^n}$. Compute a polynomial basis
$(1, U, U^2, \ldots, U^{n/2-1})$ of the subfield
$\F_{p^{n/2}}$. 
Let 
$$ L = \begin{bmatrix}
  T \\
  UT \\
  \vdots \\
  U^{n/2-1} T \\
\end{bmatrix} \mbox{ and compute } 
 M =
\begin{bmatrix}
  m_{1,1}& \ldots & m_{1,\frac{n}{2}-1}     & 1     & 0      & \ldots & 0 \\
  \vdots &        &       &       & \ddots & \ddots & \vdots \\
  \vdots &        &       &       &        & \ddots & 0 \\
  m_{\frac{n}{2}} &        &       & \ldots &        &   m_{\frac{n}{2},\frac{n}{2}-1}    & 1 \\
\end{bmatrix}$$
to be the reduced echelon form of $L$.
Then we define the lower triangular matrix made of the $n/2 \times n/2$
identity matrix with $p$ on the diagonal in the upper left quarter, the
 $n/2 \times n/2$ zero matrix in the upper right quarter, and the
 $n/2 \times n$ matrix $M$ in reduced echelon form 
in the lower half. 
Moreover, if $\deg(f) > n$, then we add $(\deg f-n-1)$ rows made of the coefficients of
$x^i\psi$ where $\F_{p^n} = \F_p[x]/(\psi(x))$, for $0 \leq i < \deg f-n-1$.
Finally we apply the LLL algorithm to this matrix.
The short vector gives us a preimage $P$ whose pseudonorm is bounded by $Q^{1/2}$, with
a polynomial selection such that $\|f\|_\infty=O(1)$ (such as conjugation or GJL).
Applying Lemma~\ref{lemma: asympt cpx boot}, we set the bound $B_1$ to
be $B_1 = L_Q[2/3, ((1/2)^2/3)^{1/3} \approx 0.436]$. The running time of
Algorithm~\ref{alg:boot-d-subfield-ExTNFS}
will be $L_q[1/3, (3/2)^{1/3} \approx 1.144]$.
We obtain preimages $P$ whose pseudonorm is bounded by $Q^{1-\frac{1}{2n}}$ with
the \JLSVi{} polynomial selection method as shown in Example~\ref{ex:Fp6}.
Applying Lemma~\ref{lemma: asympt cpx boot}, we set the bound $B_1$ to
be $B_1 = L_Q[2/3, ((1-\frac{1}{2n})^2/3)^{1/3}]$. The running time of
Algorithm~\ref{alg:boot-d-subfield-ExTNFS}
will be $L_Q[1/3, (3(1-\frac{1}{2n}))^{1/3}]$.
\end{example}

\section{Optimal representation: monic polynomial of degree $\varphi(n)$}
\label{sec:optimal-repr}

In Section~\ref{sec:general-strategy}, we exploited the largest proper
subfield $\F_{p^d}$ of $\F_{p^n}$ to find an alternative
representation of a given element $T\in \F_{p^n}$, with $n-d$ nonzero
coefficients, and $d-1$ coefficients (in $\F_p$) set to zero. The key ingredient
was to compute an expression of the form
$P = uT$, where $P$ has $d-1$ coefficients set to zero, and $u \in
\F_{p^d}$, so that we have the equality $(P/T)^{(p^n-1)/\Phi_n(p)} =
1$.
We can generalize this strategy: given an element $T$ in
the cyclotomic subgroup of $\F_{p^n}$, of order $\Phi_n(p)$, we would like to
compute an element $P \in \F_{p^n}$ such that $(P/T)^{(p^n-1)/\Phi_n(p)} =
1$ and $P$ has only $\varphi_n(p) = \deg \Phi_n(x)$ non-zero
coefficients in $\F_p$. To achieve that, we would like to compute an
expression
$$ T = u_1 u_2 \ldots u_i P,~\mbox{ where each } u_i \mbox{ is in a
  proper subfield } \F_{p^{d_i}}
\mbox{  of }\F_{p^n}~.$$
Given an element $T \in \F_{p^n}$ such that
$T^{(p^n-1)/\Phi_n(p)}\neq 1$ (in other words, its order in the cyclotomic subgroup of
$\F_{p^n}$ is not zero), we can sometimes compute an element $P$ with $\varphi(n)$ non-zero
coefficients, where $\varphi(n)$ is the Euler totient function,
plus a monic leading term. Since in
Algorithm~\ref{alg:boot-d-subfield-ExTNFS} we do not need a one-to-one 
correspondence between the given elements of the cyclotomic subgroup
on one hand, and their representation with only $\varphi(n)$ non-zero non-one
coefficients on the other hand, we can just solve a system of equations even if we do not
expect a solution at all times. If no such compact representation is
found, one picks a new $t$ and tests for the next $g^t T_0$. 
To define the system we need to solve, we list all the distinct subfields $\F_{p^d}$ of
$\F_{p^n}$ that are not themselves contained in another proper subfield,
compute a polynomial basis for each of them, and allow a degree
of freedom for the coefficients to be $\varphi(d)$ for each subfield $\F_{p^d}$.
If we consider the system as a Gr{\"o}bner basis computation, it becomes very costly even for
$\F_{p^{30}}$, where we need to handle $n-\varphi(n)-1=21$ variables.
We give a numerical example for $\F_{p^6}$.

What we do is different than what is done in XTR and CEILIDH
compact representations. In the XTR cryptosystem \cite{C:LenVer00},
the elements of the cyclotomic subgroup of
  $\F_{p^6}$ are represented with an optimal normal basis over $\F_{p^2}$, also
  in normal basis representation. Only their trace over $\F_{p^2}$ is considered
  for representation, storage, and transmission. 
In \cite{C:VanWoo04,EC:DGPRSSW05}, the aim is to define
a one-to-one correspondence between the elements in the torus of $\F_{p^n}$ and the
set of coefficients $(\F_p)^{\varphi(n)}$. This optimal compression
was achieved for $n=6$ but not for $n=30$.
  These techniques are not compatible with the representation of the elements in
  the NFS algorithm: one chooses a representation by choosing two polynomials
  $f_0, f_1$ that define the two number fields involved in the algorithm. One
  cannot change the representation afterwards: the elements in the individual
  discrete logarithm phase should be represented in the same way as the elements
  of the factor basis.

\subsection{Compressed representation of elements in the cyclotomic subgroup of
  $\F_{p^6}$ by a monic polynomial of degree 2}
\label{subsec:Groebner}

We consider the finite field $\F_{p^6}$. We will use the two subfields $\F_{p^2}$
 and $\F_{p^3}$ to cancel three coefficients.
Let $U \in \F_{p^6}$ such that $(1, U, U^2)$ is a basis
of $\F_{p^3} \subset \F_{p^6}$. Let $V \in \F_{p^6}$ such that $(1, V)$
is a basis of $\F_{p^2} \subset \F_{p^6}$.
We want to solve 
$$uvwT = (u_0 + u_1 U + u_2 U^2) (v_0 + v_1 V) w T = P~,$$ 
where $u = u_0 + u_1 U + u_2 U^2 \in \F_{p^3}$, $v = v_0 + v_1 V \in \F_{p^2}$, $w \in \Fp$, and $P \in
\F_{p^6}$ is represented by a monic polynomial in $x$ of degree 2.
To simplify, we set $u_2 = v_1 = 1$ so that we obtain equations where
we can recursively eliminate the variables by computing resultants.
We compute $u, v, w$ such that $uvwT = P$, where $P=a_0+a_1x+x^2$
is monic of degree 2.
We define the lattice
$$ L = \begin{bmatrix}
  p & 0 & 0 \\
  0 & p & 0 \\
  a_0 & a_1 & 1 \\ 
\end{bmatrix}~.$$
The determinant of $L$ is $p^2$ hence LLL$(L)$ computes a short vector
$P$ 
of coefficient size bounded by $C p^{2/3}$, where $C$ is the
LLL approximation factor (we can take $C\approx 1$ in this practical
case).
The pseudonorm of $P$ will be in the \JLSVi{} case $|\Reslt(P,f)| \approx
\|P\|_\infty^6 \|f\|_\infty^2 = p^5 = Q^{5/6}$. This is better than
the bound $Q^{11/12}$ obtained with the cubic subfield cofactor
method.
This specific method can be generalized to specific cases of finite
fields where reducing as much as possible the degree of the target is
the best strategy, as in Example~\ref{ex:Fp6-Groebner}.
This technique was implemented in \cite{SAC:GGMT17} for computing a
new discrete logarithm record in $\F_{p^6}$ of 422 bits.

\begin{example}\label{ex:Fp6-Groebner}
  We take the same finite field parameters as in
  Example~\ref{ex:Fp6}, where $\F_{p^6} = \F_p[x]/(f(x))$. 
  $g=x+3$ is a generator of $\F_{p^6}$. $(1,U,U^2)$ where $U = g^{1+p^3}$ is a basis of
  $\F_{p^3}$ and $(1,V)$ where $V= g^{1+p^2+p^4}$ is a basis of $\F_{p^2}$.
We solve the system $(u_0 + u_1 U + U^2)(v_0 + V) T = P$ where $u_i,
v_i \in \Fp$ and $P$ is monic and represented by a polynomial of degree 2
instead of 5. We ran Algorithm~\ref{alg:boot-d-subfield-ExTNFS} with
this modification on the same machine
(Intel Xeon E5-2650 at 2.0GHz with
hyperthreading turned on), from $g^0T_0$ to $g^{90000}T_0$.
On average, the set of $I$ candidates $g^iT_0$ led to six times more monic degree two
polynomials $P_i$. We found that the third polynomial
output for $T=g^{60928}T_0$ has a 64-bit-smooth pseudonorm. Testing the
90000 $g^iT_0$ (that is, $2.7\cdot 10^{5}$ pseudonorms) took 1.2
core-day:
$$
\begin{array}{@{}r@{\hspace{3pt}}c@{\hspace{3pt}}l@{}}
u &=& \coeffs{12307232765040677532260293} + \coeffs{18116887363761988927417497\ } U + U^2 \\
v &=& \coeffs{30422514788629575495025401} + V \\
w &=& \coeffs{21470888563719305004900851} \\
P &=& u v w T \\
  &=& x^2 + \coeffs{479190487430850236087613\ } x +
      \coeffs{6943966382910680737931850\ }. \\
\end{array}$$
We checked that $(P/T)^{\frac{p^6-1}{\ell}} = 1$, meaning that
$\log_g P = \log_g T = 60928 + \log_g T_0$.
Then we reduce the lattice defined by the matrix
$$\begin{bmatrix}
  p & 0 & 0 \\
  0 & p & 0 \\
  \coeffs{6943966382910680737931850} & \coeffs{479190487430850236087613} & 1 \\
\end{bmatrix} $$
to get three polynomials of smaller coefficients, the third one being
$$R = \coeffs{107301402613441938\ } x^2 - \coeffs{32014642452727111\ } x
+ \coeffs{60125316588415598}$$
whose pseudonorm is
$$\begin{array}{@{}l@{}}
\Reslt(R,f) = \\
    \coeffs{12474200655939339762647720853686893930822373172685245800138935320} \\
 \coeffs{22514918959041066623605301421497621878867497302294873400285994921}
\end{array} $$
of 429 bits, which corresponds to the estimate $\log_2 Q^{5/6} = 423$ bits.
We still have $\log_g \rho(P) \equiv \log_g T_0 + 60928 \bmod \ell$.
The pseudonorm is 64-bit-smooth, and its factorization into prime ideals is
$$\begin{array}{@{}l@{}}
 \langle \coeffs{11},x + \coeffs{8} \rangle
 \langle \coeffs{23},x + \coeffs{15} \rangle
 \langle \coeffs{12239},x + \coeffs{482} \rangle \mbox{ (small) }\\
 \langle \coeffs{1144616018827},x + \coeffs{218590032699} \rangle \\
 \langle \coeffs{2682498999539},x + \coeffs{1582479651452} \rangle\\
 \langle \coeffs{42175797334421},x + \coeffs{14828919302862} \rangle \\
 \langle \coeffs{1195156519724071},x + \coeffs{966160984838340} \rangle\\
 \langle \coeffs{13533793331200309},x + \coeffs{12224259030902272} \rangle\\
 \langle \coeffs{92644276473186311},x + \coeffs{5754482791048201} \rangle\\
 \langle \coeffs{101186915694167857},x + \coeffs{42826432866764905} \rangle\\
 \langle \coeffs{20516170632026633467},x + \coeffs{14633926248916275064} \rangle~.
\end{array}$$
The first three ideals are small enough to be in the factor basis, and
eight ideals on side 0 remain to be descended.
\end{example}

\section*{Conclusion}
The algorithms presented in this paper were implemented in Magma and used for
cryptographic-size record computations. It was shown in \cite{EPRINT:ACCMOR16}
that combined with a practical variant of Joux's  algorithm,
our Algorithm~\ref{alg:boot in small char} allows to compute a discrete logarithm in
the finite field $\F_{3^{6\cdot 709}}$ at the same cost as in
$\F_{3^{6\cdot 509}}$ with the previous state of the art.
The large characteristic variant (Algorithm~\ref{alg:boot-d-subfield-ExTNFS}) was used in \cite {SAC:GGMT17} for a 422-bit
record computation in $\F_{p^6}$. It would be interesting to
be able to generalize it further, to be able to exploit at the same
time several subfields, and provide a practical implementation of it for
cryptographic sizes.

\section*{Acknowledgments}
The author is grateful to Francisco Rodr{\'i}guez-Henr{\'i}quez,
Frederik Vercauteren, Robert Granger and Thorsten Kleinjung,
François Morain, Pierrick Gaudry, Laurent Gr\'{e}my, Luca De Feo,
and the other researchers who helped to improve this work.
All these very fruitful discussions started at the ECC 2015 conference, the
CATREL workshop and the Asiacrypt 2015 conference; 
in particular, the author would like to thank the anonymous reviewers of
Asiacrypt 2015 who suggested a generalization.

\bibliographystyle{amsplain}
\bibliography{abbrev3-short,biblio}
\end{document}